\documentclass[a4paper,10pt]{article}
\usepackage{original}
\usepackage{amsmath,amssymb}
\usepackage{braket}
\usepackage{graphicx}
\usepackage{booktabs}
\usepackage{array}
\usepackage{longtable}
\usepackage{hyperref}
\usepackage[capitalize]{cleveref}
\usepackage{enumitem}
\setlist[enumerate]{leftmargin=1.5em, labelsep=0.4em}

\usepackage{amsthm}
\theoremstyle{plain}
\newtheorem{theorem}{Theorem}
\newtheorem{lemma}{Lemma}
\newtheorem{corollary}{Corollary}
\newtheorem{proposition}{Proposition}
\theoremstyle{definition}
\newtheorem{example}{Example}
\Crefname{theorem}{Theorem}{Theorems}
\Crefname{proposition}{Proposition}{Propositions}
\Crefname{lemma}{Lemma}{Lemmas}
\Crefname{corollary}{Corollary}{Corollaries}
\Crefname{example}{Example}{Examples}

\DeclareMathOperator{\ind}{ind}
\DeclareMathOperator{\sgn}{sgn}
\DeclareMathOperator{\Ran}{Ran}
\DeclareMathOperator{\spn}{span}

\begin{document}

\title{Operator quantum singular value transformation: Spectral mapping, design and topological invariants}
\author{Chusei Kiumi}
\email{c.kiumi.qiqb@osaka-u.ac.jp}
\affiliation{Center for Quantum Information and Quantum Biology, The University of Osaka,\newline 1-2 Machikaneyama, Toyonaka, Osaka 560-0043, Japan}
\author{Yohei Tanaka}
\email{tana0035@fc.ritsumei.ac.jp}
\affiliation{Department of Mathematical Sciences, College of Science and Engineering, Ritsumeikan University,\newline 1-1-1 Noji-higashi, Kusatsu, Shiga 525-8577, Japan}
\author{Akito Suzuki}
\email{akito.suzuki@komatsu-u.ac.jp}
\affiliation{Department of Production Systems Engineering and Sciences, Komatsu University,\newline 1-3 Shicho-machi, Komatsu, Ishikawa 923-8511, Japan}

\begin{abstract}
    Quantum singular value transformation (QSVT) provides a unifying framework for quantum algorithms by applying the polynomials of quantum signal processing (QSP) to block-encoded matrices. We connect this framework with spectral analysis, design and topological protection in condensed matter physics. Extending QSVT beyond matrices to possibly infinite-dimensional operators (operator QSVT), we establish the following. First, it transforms every block of the operator by the QSP polynomials; no singular value decomposition is needed. Second, a spectral mapping theorem determines the full spectrum of the QSVT operator, including continuous spectra, from the embedded operator through the QSP polynomial, allowing spectral analysis without translation invariance. Third, the spectral mapping turns QSP into spectral design: the QSP phases can stretch, fold and separate the bands and place the eigenvalues. Finally, the kernel sectors that conventional QSVT usually sets aside carry a topological index, invariant under the choice of QSP phases, local perturbations and gap-preserving deformations, which guarantees topologically protected states. These results unify two quantum walks: Szegedy's walk from quantum algorithms and the split-step walk from condensed matter physics. Their kernel sectors encode the graph's Euler characteristic and the difference of bulk winding numbers, respectively. For the split-step walk, QSP breaks chiral symmetry but not the topological protection, and can move the protected state into a continuous band to form a bound state in the continuum (BIC). 
\end{abstract}

\section{Introduction}
Quantum singular value transformation (QSVT) \cite{QSVT,grand} provides a unifying framework for quantum algorithms, including search \cite{grover}, Hamiltonian simulation \cite{hamiltonian} and linear system solving \cite{HHL}, and serves as a foundation for modern algorithm design. It expresses a wide range of tasks through polynomial transformations of a matrix embedded in a unitary. This embedding is called a block encoding. Qubitization reduces the action of the encoding unitary to two-dimensional subspaces associated with singular values \cite{LC19}. Quantum signal processing (QSP), originating from composite-pulse control \cite{LYC16}, then alternates the unitary with tunable phase rotations to implement polynomial transformations on these subspaces \cite{LC17}. The rotation phases specify the target polynomial, and the number of alternations determines the query complexity and computational cost.

In condensed matter physics, spectral and topological properties are central to understanding transport, localisation and phase transitions, with continuous spectra describing extended bands and eigenvalues corresponding to localised states. In periodically driven (Floquet) systems, the unitary evolution over one period determines the quasi-energy spectrum, and the driving protocol engineers its band structure \cite{KitagawaBergDemlerRudner,RudnerLindnerBergLevin,OkaKitamuraFloquet}. Topological invariants, in turn, connect bulk properties to the existence of boundary or interface states and their stability under perturbations \cite{BES94,HasanKane,KaneLubensky,ProdanSchulzBaldes}.

Despite their different origins, these two settings share a common structure: QSVT is a driven unitary evolution where the QSP phases act as a driving protocol going beyond periodic Floquet driving. This raises a fundamental question: can we harness the powerful polynomial control of QSVT to flexibly engineer band structures and bound states while preserving topological protection? Conventional QSVT, however, is formulated for finite-dimensional matrices and focuses on transforming the encoded block, obscuring continuous bands and topological structure.

We develop \emph{operator QSVT}, extending QSVT to the full unitary evolution on possibly infinite-dimensional Hilbert spaces and connecting block encoding, qubitization and QSP to spectral analysis, spectral design and topological protection.
First, we show that every block of the QSVT operator is transformed directly by the QSP polynomials at the operator level without requiring a singular value decomposition, recovering standard QSVT as an immediate corollary.
Second, we establish a spectral mapping theorem by decomposing the Hilbert space into the qubitized sector and the kernel sectors, determining the full spectrum of the QSVT operator---including continuous spectra---from the encoded operator through the QSP polynomial and allowing spectral analysis of the full unitary from a smaller operator without assuming translation invariance.
Third, we turn this spectral mapping into a design method using the established QSP theory of realisable polynomials \cite{LYC16,Haah19,DongMengWhaleyLin,YingQSP,MotlaghWiebe}. By choosing the QSP phases, this composite-pulse mechanism can stretch, fold and separate spectral bands and place kernel eigenvalues, with explicit constraints and degree bounds.
Finally, whereas the qubitized sector obeys a conjugation symmetry, the kernel sectors that conventional QSVT usually sets aside can break this symmetry. This imbalance is measured by topological indices, which guarantee eigenstates robust against the QSP phases, compact perturbations and gap-preserving deformations. This protection requires no gap around the eigenvalues of the protected states, which may lie inside continuous bands.

Quantum walks provide a concrete setting for connecting quantum algorithms and topological physics. Szegedy's construction \cite{szegedy} embeds the discriminant of a classical Markov chain in a unitary and underlies several quantum algorithms \cite{montanaro,pagerank,QFF,Unified}. The split-step walk, in turn, models topological phases and localised interface states \cite{Kitagawa,kitagawa2}. Operator QSVT unifies both walks, and their topological indices give the Euler characteristic for Szegedy's walk on finite graphs and the difference of bulk winding numbers for the split-step walk. For the split-step walk, generic QSP phases break chiral symmetry while preserving the topological index, which guarantees a localised state whose eigenvalue can be controlled by the phases. A degree-two polynomial places this eigenvalue inside an absolutely continuous band, producing a bound state in the continuum (BIC) \cite{BICreview} that survives changes of the coin rotations on finitely many sites.

Our framework connects to several lines of prior work. On the algorithmic side, infinite-dimensional extensions of QSP have focused on infinite phase sequences via nonlinear Fourier analysis \cite{DongInfiniteQSP,AlexisNLFA} and continuous-variable bosonic systems \cite{SinananSinghCVQSP}, rather than the full unitary evolution of general bounded block-encoded operators. In the quantum-walk literature, spectral mapping and topological index theory have established foundations. For walks built from two reflections, spectral mapping theorems relate their spectra to those of self-adjoint discriminant operators, including on infinite graphs \cite{SMT}. This has been applied to Grover walks on crystal lattices \cite{crystal}, homological trapping \cite{MachidaSegawa}, periodicity \cite{periodicity17,kubota19,kubota21}, generators and long-time behaviour \cite{GJS,Generator}, and the spectra and bound states of split-step walks \cite{SS1,SS2}. Spectral mappings have also been obtained for coins with two distinct eigenvalues combined with a flip-flop shift \cite{TwoEigenCoin}, and extended to non-unitary walks \cite{NonunitarySMT}. Similarly, the staggered walk with Hamiltonians is a product of exponentials of two reflections \cite{StaggeredHamiltonians}. Topological protection likewise has an established index-theoretic foundation: for chiral symmetric walks, an index guarantees bound states robust against compact perturbations \cite{SuzukiSUSY}, with explicit classifications and index theorems for one-dimensional split-step and asymptotically periodic walks \cite{SuzukiTanakaWitten,MatsuzawaIndex,MTW,MatsuzawaWitten}, and bulk--edge correspondence relates bulk index differences to protected interface states \cite{CGSVWW,CGGSVWW}. The symmetry indices of a chiral symmetric walk are expressed in terms of Fredholm indices of its chiral blocks \cite{CGSVW18}, or of the blocks of a half-step unitary \cite{CGWW21}. Operator QSVT brings these two lines together. Both concern the chiral symmetric case of two reflections; we extend the spectral mapping theorem of \cite{SMT} and the index-based protection of \cite{CGSVW18,CGWW21} to arbitrary QSP phases, where the mapping becomes the QSP polynomial, chiral symmetry can be broken, and the protected eigenvalues move away from $\pm1$.

\subsection*{The standard formulation of QSVT}

We first recall the finite-dimensional formulation of QSVT \cite{QSVT}, following the exposition of \cite{grand}, and introduce notation that will also be used for its operator extension. Let $U$ be a unitary operator on a finite-dimensional Hilbert space $\mathcal{H}$ and let $\Pi _{\mathfrak{H}}$ and $\Pi _{\mathfrak{K}}$ be orthogonal projections on $\mathcal{H}$. These data define a projected unitary encoding, referred to here as a block encoding, of
\begin{equation}\label{eq:blockenc}
    A:=\Pi _{\mathfrak{K}} U\Pi _{\mathfrak{H}} ,
\end{equation}
a contraction from $\mathfrak{H} :=\Ran \Pi _{\mathfrak{H}}$ to $\mathfrak{K} :=\Ran \Pi _{\mathfrak{K}}$, so that $A:\mathfrak{H}\rightarrow \mathfrak{K}$ and $\mathcal{H} =\mathfrak{H} \oplus \mathfrak{H}^{\perp } =\mathfrak{K} \oplus \mathfrak{K}^{\perp }$. Phases are inserted through rotations generated by the reflections $2\Pi_{\mathfrak H}-I$ and $2\Pi_{\mathfrak K}-I$:
\begin{equation}
    R_{\mathfrak{H}}( \varphi ) :=e^{i\varphi \left( 2\Pi _{\mathfrak{H}} -I\right)} ,\qquad
    R_{\mathfrak{K}}( \varphi ) :=e^{i\varphi \left( 2\Pi _{\mathfrak{K}} -I\right)} .
\end{equation}
For an integer $d\geqslant1$ and a phase vector $\Phi \in \mathbb{R}^{d}$, the phased alternating sequence of \cite{QSVT} is, for even $d$,
\begin{equation}\label{eq:UPhi}
    U_{\Phi } =\prod _{j=1}^{d/2}\left( R_{\mathfrak{H}}( \varphi _{2j-1}) U^{*} R_{\mathfrak{K}}( \varphi _{2j}) U\right) .
\end{equation}
Here products are written in left-to-right index order. For odd $d$, the corresponding sequence is
\begin{equation}\label{eq:UPhi-odd}
    U_\Phi=R_{\mathfrak K}(\varphi_1)U
    \prod_{j=1}^{(d-1)/2}
    \bigl(R_{\mathfrak H}(\varphi_{2j})U^*
          R_{\mathfrak K}(\varphi_{2j+1})U\bigr).
\end{equation}
Qubitization \cite{LC19} uses singular vectors of $A$ to construct paired two-dimensional source and target spaces for each singular value $\sigma\in(0,1)$. Between their respective bases, $U$ is represented by the reflection
\begin{equation}\label{eq:Rsigma}
    R(\sigma)=\begin{bmatrix}
        \sigma&\sqrt{1-\sigma^2}\\
        \sqrt{1-\sigma^2}&-\sigma
    \end{bmatrix}.
\end{equation}
The phase rotations on the source and target sides are represented by $e^{i\varphi Z}$, where $Z=\operatorname{diag}(1,-1)$. The even sequence acts within a source space, whereas the odd sequence maps a source space to its target partner. In these bases the reflection convention for QSP gives
\begin{equation}\label{eq:QSPseq}
    S_\Phi(x):=\prod_{j=1}^d\bigl(e^{i\varphi_jZ}R(x)\bigr)
    =\begin{cases}
    \begin{bmatrix}
        P_\Phi(x)&iQ_\Phi(x)\sqrt{1-x^2}\\
        iQ_\Phi^*(x)\sqrt{1-x^2}&P_\Phi^*(x)
    \end{bmatrix},&d\text{ even},\\[1em]
    \begin{bmatrix}
        P_\Phi(x)&Q_\Phi(x)\sqrt{1-x^2}\\
        Q_\Phi^*(x)\sqrt{1-x^2}&-P_\Phi^*(x)
    \end{bmatrix},&d\text{ odd}.
    \end{cases}
\end{equation}
Here $P_\Phi,Q_\Phi\in\mathbb C[x]$ are polynomials with complex coefficients in a real variable $x$, determined by $\Phi$, and $P_\Phi^*$ and $Q_\Phi^*$ denote coefficientwise complex conjugation. In both cases $\deg P_\Phi\leqslant d$, $\deg Q_\Phi\leqslant d-1$, the parities are $d$ and $d-1$ modulo $2$, and
\[
    |P_\Phi(x)|^2+(1-x^2)|Q_\Phi(x)|^2=1\qquad(x\in[-1,1]).
\]
These properties follow by multiplication and unitarity. 

Conversely, when only the real part of the upper-left entry is prescribed, every real polynomial of the appropriate parity, degree at most $d$, and modulus at most one on $[-1,1]$ can be realised as $\operatorname{Re}P_\Phi$, with a companion $Q_\Phi$ supplied by the construction \cite{QSVT}.

For a singular value decomposition $A=\sum_i\sigma_i\ket{\tilde\psi_i}\bra{\psi_i}$, the encoded block is
\begin{equation}\label{eq:qsvt-svd}
    \begin{cases}
    \Pi_{\mathfrak H}U_\Phi\Pi_{\mathfrak H}
       =\sum_i P_\Phi(\sigma_i)\ket{\psi_i}\bra{\psi_i},&d\text{ even},\\
    \Pi_{\mathfrak K}U_\Phi\Pi_{\mathfrak H}
       =\sum_i P_\Phi(\sigma_i)\ket{\tilde\psi_i}\bra{\psi_i},&d\text{ odd},
    \end{cases}
\end{equation}
where the right singular vectors include an orthonormal basis of $\ker A$ in the even case \cite{QSVT}. In infinite dimensions, however, a general operator $A$ does not admit such a singular value decomposition.

\subsection*{Summary of the results}

We now allow $\mathcal H$ to be infinite dimensional and write $U$ in block form from $\mathfrak H\oplus\mathfrak H^\perp$ to $\mathfrak K\oplus\mathfrak K^\perp$:
\begin{equation}\label{eq:blocks}
    U=\begin{bmatrix}
        A & B \\
        C & D
    \end{bmatrix} ,\qquad
    \begin{array}{ll}
        A:\mathfrak{H}\rightarrow \mathfrak{K} , & B:\mathfrak{H}^{\perp }\rightarrow \mathfrak{K} ,\\
        C:\mathfrak{H}\rightarrow \mathfrak{K}^{\perp } , & D:\mathfrak{H}^{\perp }\rightarrow \mathfrak{K}^{\perp } .
    \end{array}
\end{equation}
Here $A$ is the contraction encoded in \eqref{eq:blockenc}, with $\|A\|\leqslant1$ by unitarity of $U$. Even when $\mathfrak H=\mathfrak K$, no self-adjointness or normality is assumed. For phases $\theta _{k} ,\phi _{k} \in [ 0,2\pi )$ we set
\begin{equation}\label{eq:Wk-intro}
    W_{k} :=R_{\mathfrak{H}}( \theta _{k}) U^{*} R_{\mathfrak{K}}( \phi _{k}) U,
    \qquad
    W^{( n)} :=W_{n} W_{n-1} \cdots W_{1} ,
\end{equation}
which is $U_{\Phi }$ of \eqref{eq:UPhi} for $d=2n$ and $\Phi=(\theta_n,\phi_n,\ldots,\theta_1,\phi_1)$ in the left-to-right product convention; we call it the \emph{even QSVT iterate}. Since $P_\Phi$ is even and $Q_\Phi$ is odd, we define $\Pi_n,\Phi_n\in\mathbb C[X]$ by
\[
    P_\Phi(x)=\Pi_n(x^2),\qquad
    Q_\Phi(x)=x\Phi_n(x^2)\qquad(x\in\mathbb R).
\]
Here $\mathbb C[X]$ denotes the set of polynomials with complex coefficients, evaluated at a real number $x$ or at a self-adjoint operator $X$. We write $AA^*|_{(0,1)}$ for the restriction of $AA^*$ to the spectral subspace $E_{AA^*}((0,1))\mathfrak K$, where $E_{AA^*}$ is the spectral measure of $AA^*$. Define the four phases
\[
    \omega_A:=\prod_{k=1}^{n}e^{i(\theta_k-\phi_k)},\quad
    \omega_B:=\prod_{k=1}^{n}e^{-i(\theta_k+\phi_k)},\quad
    \omega_{A^*}:=\overline{\omega_A},\quad
    \omega_{B^*}:=\overline{\omega_B} .
\]

Our results are the following.

\begin{longtable}{@{}>{\raggedright\arraybackslash}p{0.18\linewidth}p{0.76\linewidth}@{}}
    \toprule
    \endhead
    \textbf{\Cref{thm:poly}}\newline Polynomial transformation &
    Every block of the QSVT operator $W^{(n)}$ has an explicit polynomial formula in terms of the blocks $A,B,C,D$ and the QSP polynomials $\Pi_n,\Phi_n$. The derivation uses only block-operator multiplication, without relying on finite-dimensionality or a singular value or eigenvalue decomposition.\\
    \addlinespace
    \textbf{\Cref{thm:SMT}}\newline Spectral mapping &
    With $T$ ranging over $A,A^*,B,B^*$, the spectrum of $W^{(n)}$ is determined by
\begin{equation*}
    \sigma(W^{(n)})=\left\{z\in\mathbb T:\operatorname{Re}z\in\operatorname{Re}\Pi_n\bigl(\sigma(AA^*|_{(0,1)})\bigr)\right\}\cup\left\{\omega_T:\ker T\neq\{0\}\right\}.
\end{equation*}
The theorem determines the essential and point spectra, as well as the pure point, absolutely continuous and singular continuous spectral components.\\
    \addlinespace
    \textbf{\Cref{thm:design}}\newline Spectral design &
    The QSP phases control both spectral bands and eigenvalues. Any input band can be mapped by $\operatorname{Re}\Pi_n$ onto any prescribed interval $[m,M]\subseteq[-1,1]$. The number of bands cannot increase. Explicit degree bounds quantify the cost of stretching a band, separating two bands, and separating a kernel eigenvalue from a band. Embedding an initially separated kernel eigenvalue in an absolutely continuous band takes exactly two steps. \\
    \addlinespace
    \textbf{\Cref{thm:index}}\newline Topological invariant &
    The Fredholm indices of $A,B$ measure imbalance between kernel sectors:
\[
    \ind A=\dim\ker A-\dim\ker A^*,\qquad
    \ind B=\dim\ker B-\dim\ker B^*.
\]
A non-zero index guarantees eigenstates topologically protected
against changes in the QSP phase sequence $\Phi$ and
compact perturbations and gap-preserving deformations of
$U,\Pi_{\mathfrak H},\Pi_{\mathfrak K}$ under the theorem's conditions. \\
    \bottomrule
\end{longtable}

In this framework, Szegedy's walk on a finite graph $G$ has the Markov-chain discriminant as $A$ and $(\ind A,\ind B)=(0,-2\chi(G))$ in terms of the Euler characteristic $\chi(G)$. For the split-step walk, $AA^*$ is a Jacobi operator, $\ind A$ and $\ind B$ are differences of bulk winding numbers, and when $\ind A\neq0$ a two-step polynomial $\operatorname{Re}\Pi_2$ folds the continuous band over the protected eigenvalue $\omega_A$ or $\omega_{A^*}$ to produce a BIC.

The remainder of this paper is organised as follows. \Cref{sec:main} develops operator QSVT and presents the four main results in \cref{sec:poly,sec:SMT,sec:design,sec:symmetry}, respectively. \Cref{sec:example} treats Szegedy's and split-step quantum walks within this framework. \Cref{sec:conclusion} discusses conclusions and future directions. The proofs are collected in the appendices.

\section{Theory of operator QSVT\label{sec:main}}

Let $\mathcal H$ be a (possibly infinite-dimensional) Hilbert space and let $\Pi_{\mathfrak H},\Pi_{\mathfrak K}$ be orthogonal projections with ranges $\mathfrak H:=\Ran\Pi_{\mathfrak H}$ and $\mathfrak K:=\Ran\Pi_{\mathfrak K}$, so that $\mathcal H=\mathfrak H\oplus\mathfrak H^\perp=\mathfrak K\oplus\mathfrak K^\perp$. As in \eqref{eq:blocks}, we write the unitary $U$ on $\mathcal H$ in block form from $\mathfrak H\oplus\mathfrak H^\perp$ to $\mathfrak K\oplus\mathfrak K^\perp$ as
\[
    U=\begin{bmatrix} A & B \\ C & D \end{bmatrix},
\]
where the unitarity conditions $U^*U=I_{\mathcal H}$ and $UU^*=I_{\mathcal H}$ read
\begin{equation}
    \begin{aligned}
        A^{*} A+C^{*} C & =I_{\mathfrak{H}} , & B^{*} B+D^{*} D & =I_{\mathfrak{H}^{\perp }} , & A^{*} B+C^{*} D & =0, & B^{*} A+D^{*} C & =0,\\
        AA^{*} +BB^{*} & =I_{\mathfrak{K}} , & CC^{*} +DD^{*} & =I_{\mathfrak{K}^{\perp }} , & AC^{*} +BD^{*} & =0, & CA^{*} +DB^{*} & =0.
    \end{aligned}
\end{equation}
With the phase rotations $R_{\mathfrak H}(\theta)=e^{i\theta(2\Pi_{\mathfrak H}-I)}$ and $R_{\mathfrak K}(\phi)=e^{i\phi(2\Pi_{\mathfrak K}-I)}$ for $\theta_k,\phi_k\in[0,2\pi)$, the single step and $n$-step even QSVT iterate of \eqref{eq:Wk-intro} are
\[
    W_k:=R_{\mathfrak H}(\theta_k)U^*R_{\mathfrak K}(\phi_k)U,\qquad
    W^{(n)}:=W_nW_{n-1}\cdots W_1.
\]

\subsection{Polynomial transformation of the embedded operator\label{sec:poly}}

The iterate $W^{( n)} =W_{n} \cdots W_{1}$ transforms the self-adjoint positive operators $A^{*} A$ and $D^{*} D$ by polynomials. We treat the even iterate $W^{( n)}$ of \eqref{eq:UPhi} and the odd one $R_{\mathfrak{K}}( \phi _{n+1}) UW^{( n)}$ of \eqref{eq:UPhi-odd} together. Here $\mathbb{C}[ X]$ is the set of polynomials with complex coefficients, evaluated at a real number $x$ or at a self-adjoint operator $X$, and for $P\in\mathbb{C}[X]$ we write $P^{*}$ for the polynomial obtained by conjugating the coefficients of $P$, so that $P^{*}(x)=\overline{P(x)}$ for real $x$.

Direct block multiplication using the unitarity relations gives a single step $W_k$ in the form
\begin{equation}\label{eq:Wk}
    W_{k} =\begin{bmatrix}
        P_{k}\left( A^{*} A\right) & iQ_{k}\left( A^{*} A\right) A^{*} B \\
        iQ_{k}^{*}\left( D^{*} D\right) B^{*} A & P_{k}^{*}\left( D^{*} D\right)
    \end{bmatrix}
    \quad \text{on } \mathfrak{H} \oplus \mathfrak{H}^{\perp } ,
\end{equation}
where $P_{k} ,Q_{k} \in \mathbb{C}[ X]$ are given by
\begin{equation}\label{eq:PQ}
    P_{k}( x) :=e^{i\theta _{k}}\left( e^{-i\phi _{k}} +2i\sin \phi _{k} x\right) ,\qquad Q_{k}( x) :=2e^{i\theta _{k}}\sin \phi _{k} .
\end{equation}
Note that $Q_{k}$ is in fact a constant; we nevertheless write it as a polynomial, because the form \eqref{eq:Wk} is what is preserved under iteration, as we now show.
\begin{theorem}[Polynomial transformation]\label{thm:poly}
    Let $n\geqslant 1$.
    \begin{enumerate}
        \item[(a)] \emph{Even iteration.} There are $\Pi _{n} ,\Phi _{n} \in \mathbb{C}[ X]$ such that, on $\mathfrak{H} \oplus \mathfrak{H}^{\perp }$,
              \begin{equation*}
                  W^{( n)} =\begin{bmatrix}
                      \Pi _{n}\left( A^{*} A\right) & i\Phi _{n}\left( A^{*} A\right) A^{*} B \\
                      i\Phi _{n}^{*}\left( D^{*} D\right) B^{*} A & \Pi _{n}^{*}\left( D^{*} D\right)
                  \end{bmatrix} .
              \end{equation*}
              With $\Pi _{0} \equiv 1$ and $\Phi _{0} \equiv 0$, the pair obeys the recursion
              \begin{equation*}
                  \begin{bmatrix}
                      \Pi _{n}( x) \\
                      \Phi _{n}^{*}( x)
                  \end{bmatrix} =\begin{bmatrix}
                      P_{n}( x) & -Q_{n}( x)( 1-x) x \\
                      Q_{n}^{*}( x) & P_{n}^{*}( x)
                  \end{bmatrix}\begin{bmatrix}
                      \Pi _{n-1}( x) \\
                      \Phi _{n-1}^{*}( x)
                  \end{bmatrix} .
              \end{equation*}
        \item[(b)] \emph{Odd iteration.} From $\mathfrak{H} \oplus \mathfrak{H}^{\perp }$ to $\mathfrak{K} \oplus \mathfrak{K}^{\perp }$,
              \begin{equation*}
                  R_{\mathfrak{K}}( \phi _{n+1}) UW^{( n)} =\begin{bmatrix}
                      \Theta _{n}\left( AA^{*}\right) A & \Omega _{n}\left( AA^{*}\right) B \\
                      \Omega _{n}^{*}\left( DD^{*}\right) C & \Theta _{n}^{*}\left( DD^{*}\right) D
                  \end{bmatrix},
              \end{equation*}
              where
              \begin{align*}
                   & \Theta _{n}( x) =e^{i\phi _{n+1}}\left( \Pi _{n}( x) +i\Phi _{n}^{*}( x)( 1-x)\right) , \\
                   & \Omega _{n}( x) =e^{i\phi _{n+1}}\left( \Pi _{n}^{*}( x) +i\Phi _{n}( x) x\right) .
              \end{align*}
        \item[(c)] \emph{QSP correspondence.} Let $\Phi =( \theta _{n} ,\phi _{n} ,\dotsc ,\theta _{1} ,\phi _{1})$ and $\Phi ' =( \phi _{n+1} ,\theta _{n} ,\phi _{n} ,\dotsc ,\theta _{1} ,\phi _{1})$. Then, for all $\sigma\in\mathbb R$,
              \begin{equation}\label{eq:qsp-pair}
                  \begin{aligned}
                      P_{\Phi }( \sigma ) &=\Pi _{n}\left( \sigma ^{2}\right) , & Q_{\Phi }( \sigma ) &=\sigma \Phi _{n}\left( \sigma ^{2}\right) , \\
                      P_{\Phi '}( \sigma ) &=\sigma \Theta _{n}\left( \sigma ^{2}\right) , & Q_{\Phi '}( \sigma ) &=\Omega _{n}\left( \sigma ^{2}\right) ,
                  \end{aligned}
              \end{equation}
              where the polynomials on the left are those of the reflection sequence \eqref{eq:QSPseq} with $d=2n$ and $d=2n+1$, respectively.
    \end{enumerate}
\end{theorem}
\begin{proof}
    See \cref{app:poly}.
\end{proof}

Part (c) identifies $(\Pi_n,\Phi_n)$ with the quantum signal processing polynomials of the finite-dimensional theory. What \cref{thm:poly} adds is the route to them: they are obtained from an identity for the whole operator, without choosing a basis or constructing two-dimensional invariant subspaces, and without assuming finite-dimensionality or compactness of $A$. In particular, the QSP normalisation condition for $(P_\Phi,Q_\Phi)$, written in the variable $x=\sigma^2$, becomes
\begin{equation}\label{eq:unitarity}
    |\Pi _{n}( x)|^{2} +x( 1-x)|\Phi _{n}( x)|^{2} =1 ,\qquad x\in \mathbb{R},
\end{equation}
valid for every $n\geqslant 0$, so $|\Pi _{n}( x)| \leqslant 1$ and hence $|\operatorname{Re}\Pi _{n}( x)| \leqslant 1$ for every $x\in [ 0,1]$. When singular vectors exist, the identities of (a) and (b) recover the usual singular value transformation as follows.

\begin{corollary}[Singular value transformation]\label{cor:svt}
    Let $\sigma \in [ 0,1]$ and let $\psi \in \mathfrak{H}$ and $\tilde{\psi } \in \mathfrak{K}$ satisfy $A\psi =\sigma \tilde{\psi }$ and $A^{*} \tilde{\psi } =\sigma \psi$. Then, with $P_{\Phi }$ and $P_{\Phi '}$ as in \cref{thm:poly}(c):
    \begin{enumerate}
        \item[(a)] \emph{Even iteration.}
              \begin{equation*}
                  \Pi _{\mathfrak{H}} W^{( n)} \Pi _{\mathfrak{H}} \psi =P_{\Phi }( \sigma )\,\psi .
              \end{equation*}
        \item[(b)] \emph{Odd iteration.}
              \begin{equation*}
                  \Pi _{\mathfrak{K}} R_{\mathfrak{K}}( \phi _{n+1}) UW^{( n)} \Pi _{\mathfrak{H}} \psi =P_{\Phi '}( \sigma )\,\tilde{\psi } .
              \end{equation*}
    \end{enumerate}
    In particular, when $A$ admits a singular value decomposition $A=\sum_i\sigma_i\ket{\tilde\psi_i}\bra{\psi_i}$, they recover \eqref{eq:qsvt-svd}.
\end{corollary}
\begin{proof}
    Since $A^{*} A\psi =\sigma ^{2} \psi$ and $AA^{*} \tilde{\psi } =\sigma ^{2} \tilde{\psi }$, the upper-left blocks of \cref{thm:poly}(a) and (b) together with \cref{thm:poly}(c) give $\Pi _{\mathfrak{H}} W^{( n)} \Pi _{\mathfrak{H}} \psi =\Pi _{n}( \sigma ^{2}) \psi =P_{\Phi }( \sigma ) \psi$ and $\Pi _{\mathfrak{K}} R_{\mathfrak{K}}( \phi _{n+1}) UW^{( n)} \Pi _{\mathfrak{H}} \psi =\Theta _{n}( AA^{*}) A\psi =\sigma \Theta _{n}( \sigma ^{2}) \tilde{\psi } =P_{\Phi '}( \sigma ) \tilde{\psi }$.
\end{proof}

In infinite dimensions, eigenvectors alone do not describe the full operator in general. The next subsection uses spectral measures to include continuous spectra.

\subsection{Spectral mapping theorem\label{sec:SMT}}

From here on we restrict attention to the even iterate $W^{(n)}$: whereas the odd iterate \eqref{eq:UPhi-odd} maps $\mathfrak H\oplus\mathfrak H^{\perp}$ to $\mathfrak K\oplus\mathfrak K^{\perp}$, each step of $W^{(n)}$ returns to $\mathfrak H\oplus\mathfrak H^{\perp}$, so its spectrum on $\mathcal H$ is determined entirely by $AA^{*}$ and the kernel sectors. In this subsection, we describe the spectrum of $W^{(n)}$ in terms of the spectrum of \(AA^*\). Since \(A\) is a contraction, \(AA^*\) is a bounded self-adjoint operator satisfying $0\leqslant AA^*\leqslant I_{\mathfrak K}$, and in particular $\sigma(AA^*)\subset[0,1].$ We write \(\mathbb T:=\{z\in\mathbb C\mid |z|=1\}\) for the unit circle. Between a spectral point $x\in(0,1)$ of $AA^*$ and a spectral point $z=e^{i\lambda}\in\mathbb T$ of $W^{(n)}$, the QSVT spectral mapping is governed by the relation $\cos\lambda=\operatorname{Re}\Pi_n(x)$, while the endpoints $x=0,1$ are treated separately through the kernel sectors.

Define
\begin{equation*}
    \mathcal{L} :=\overline{\left\{\begin{bmatrix}
                A^{*} f \\
                B^{*} g
            \end{bmatrix} \in \mathfrak{H} \oplus \mathfrak{H}^{\perp }\middle| f,g\in \mathfrak{K}\right\}}
    =\overline{\Ran A^{*}} \oplus \overline{\Ran B^{*}} .
\end{equation*}
The subspace $\mathcal L$ is invariant under $W^{(n)}$. Indeed, since $P(A^{*}A)A^{*}=A^{*}P(AA^{*})$ and $P(D^{*}D)B^{*}=B^{*}P(I-BB^{*})=B^{*}P(AA^{*})$ for every polynomial $P$, the block form of \cref{thm:poly} gives
\begin{equation}\label{eq:W-on-L}
    W^{( n)}\begin{bmatrix}
        A^{*} f \\
        B^{*} g
    \end{bmatrix} =\begin{bmatrix}
        A^{*}\left( \Pi _{n}\left( AA^{*}\right) f+i\Phi _{n}\left( AA^{*}\right) BB^{*} g\right) \\
        B^{*}\left( i\Phi _{n}^{*}\left( AA^{*}\right) AA^{*} f+\Pi _{n}^{*}\left( AA^{*}\right) g\right)
    \end{bmatrix} ,\qquad f,g\in\mathfrak K ,
\end{equation}
so $W^{(n)}$ maps $\mathcal L$ into itself: on $\mathcal L$ the iterate acts through functions of $AA^{*}$.

Even when $A$ and $B$ are not compact, so that no discrete singular value decomposition is available, $\mathcal{L}$ is where the singular value structure lives and carries the whole of the transformation. We split it according to the spectrum of $AA^{*}$. Let \(E_{AA^*}\) be the spectral measure of \(AA^*\), and define
\[
    \mathfrak K_0:=\Ran E_{AA^*}(\{0\}),\qquad
    \mathfrak K_{(0,1)}:=\Ran E_{AA^*}((0,1)),\qquad
    \mathfrak K_1:=\Ran E_{AA^*}(\{1\}),
\]
so that $\mathfrak K=\mathfrak K_0\oplus\mathfrak K_{(0,1)}\oplus\mathfrak K_1$, and write $AA^{*}|_{( 0,1)} :=AA^{*}|_{\mathfrak K_{( 0,1)}}$ for the non-endpoint part. The two endpoint spaces are kernels:
\[
    \mathfrak K_0=\ker AA^*=\ker A^*,
    \qquad
    \mathfrak K_1=\ker( I-AA^{*}) =\ker B^* .
\]
Indeed, $\| A^{*} f\|^{2} =\langle f,AA^{*} f\rangle$ shows that $AA^{*} f=0$ if and only if $A^{*} f=0$, and since $BB^{*} =I_{\mathfrak K} -AA^{*}$ the same argument applied to $B$ shows that $AA^{*} f=f$ if and only if $B^{*} f=0$. The same two identities show that $B^{*}$ is isometric on $\mathfrak K_{0}$ and $A^{*}$ is isometric on $\mathfrak K_{1}$. Accordingly we define
\[
    \mathcal L_0:=B^*\mathfrak K_0,
    \qquad
    \mathcal L_{(0,1)}:=\overline{A^*\mathfrak K_{(0,1)}}\oplus\overline{B^*\mathfrak K_{(0,1)}},
    \qquad
    \mathcal L_1:=A^*\mathfrak K_1,
\]
all three being closed subspaces of $\mathfrak{H} \oplus \mathfrak{H}^{\perp }$. At the two endpoints one of the two components is absent, because $A^{*}$ vanishes on $\mathfrak K_{0}$ and $B^{*}$ vanishes on $\mathfrak K_{1}$.

To describe $W^{( n)}$ on $\mathcal L_{( 0,1)}$ we use the polar decompositions
\[
    A^*=V_A(AA^*)^{1/2},\qquad
    B^*=V_B(I-AA^*)^{1/2}
\]
restricted to \(\mathfrak K_{(0,1)}\); equivalently, formally, $V_A=A^*(AA^*)^{-1/2}$ and $V_B=B^*(I-AA^*)^{-1/2}$ on the non-endpoint part. Define $J:\mathfrak K_{(0,1)}\oplus\mathfrak K_{(0,1)}\longrightarrow\mathcal L_{(0,1)}$ by
\[
    J
    \begin{bmatrix}
        f\\
        g
    \end{bmatrix}
    =
    \begin{bmatrix}
        V_A f\\
        V_B g
    \end{bmatrix}.
\]
Since \(V_A\) and \(V_B\) are isometries on \(\mathfrak K_{(0,1)}\) with ranges $\overline{A^*\mathfrak K_{(0,1)}}$ and $\overline{B^*\mathfrak K_{(0,1)}}$, and the two components lie in the orthogonal subspaces \(\mathfrak H\) and \(\mathfrak H^\perp\), we have $\|V_Af\|^2+\|V_Bg\|^2=\|f\|^2+\|g\|^2$, so that \(J\) is a unitary onto $\mathcal L_{(0,1)}$.

\begin{proposition}[Operator qubitization]\label{prop:L-decomp}
    The subspace $\mathcal{L}$ has the orthogonal decomposition \[\mathcal{L}=\mathcal{L}_{0}\oplus\mathcal{L}_{(0,1)}\oplus\mathcal{L}_{1},\] with endpoint subspaces
    $\mathcal{L}_{0} =\ker D,\,\mathcal{L}_{1} =\ker C,$
    and orthogonal complement $\mathcal{L}^{\perp } =\ker A\oplus \ker B.$
    Consequently, the Hilbert space decomposes orthogonally into $W^{(n)}$-invariant subspaces as
    \[
        \mathcal H
        =
        \mathcal L_0\oplus\mathcal L_{(0,1)}\oplus\mathcal L_1\oplus\mathcal L^{\perp }
        =
        \mathcal L_{(0,1)}\oplus \bigoplus_{T\in\{A,B,C,D\}}\ker T .
    \]
    Moreover,
    \[
        J^{*}\, W^{( n)}|_{\mathcal L_{( 0,1)}}\, J=M_{n}( AA^{*}|_{(0,1)})
        \quad\text{on }\mathfrak K_{( 0,1)} \oplus \mathfrak K_{( 0,1)} ,
    \]
     where
    \begin{equation}\label{eq:qsp-fibre}
        M_n(x):=
        \begin{bmatrix}
            \Pi_n(x)&
            i\Phi_n(x)\sqrt{x(1-x)}\\
            i\Phi_n^*(x)\sqrt{x(1-x)}&
            \Pi_n^*(x)
        \end{bmatrix} .
    \end{equation}
    Here $M_n(\sigma^2)=S_\Phi(\sigma)$ is the QSP matrix \eqref{eq:QSPseq} for $\Phi=(\theta_n,\phi_n,\dots,\theta_1,\phi_1)$ and $\sigma\in[0,1]$.
\end{proposition}
\begin{proof}
    See \cref{app:qubitization}.
\end{proof}

Qubitization builds, for each singular value $\sigma \in ( 0,1)$, the two-dimensional invariant subspace $\mathcal{L}_{\sigma } :=\spn\{ A^{*} f_{\sigma } /\sigma ,\ B^{*} f_{\sigma } /\sqrt{1-\sigma ^{2}}\} =J( \spn\{ f_{\sigma }\} \oplus \spn\{ f_{\sigma }\})$ (where $AA^*f_\sigma=\sigma^2f_\sigma$), on which $W^{( n)}$ acts as $M_{n}( \sigma ^{2})$, and reads the transformation off from it. This is the qubitization route to \cref{thm:poly}(c); the point of \cref{sec:poly} is that no such fibre is needed to reach it. When $AA^{*}$ has pure point spectrum---in particular whenever $\mathcal{H}$ is finite dimensional---the $\mathcal{L}_{\sigma }$ exhaust $\mathcal{L}_{( 0,1)}$, and \cref{thm:SMT} below reduces to a statement about eigenvalues. For a general $AA^{*}$ they do not, which is why the spectral measure is used instead. At the endpoints $\sigma\in\{0,1\}$---corresponding to the four kernel sectors $\mathcal{L}_{0}$, $\mathcal{L}_{1}$, $\ker A$ and $\ker B$---one of the two basis vectors of $\mathcal{L}_{\sigma}$ vanishes, and $W^{(n)}$ acts instead as a scalar phase on each sector, as detailed in \cref{thm:SMT}(a) below. The decomposition into the four kernel sectors and their orthogonal complement goes back to the spectral mapping theorem of \cite{SMT}; the use of block polar decompositions also appears for half-step operators of chiral symmetric walks in \cite{CGWW21}, where it is used to construct homotopies.

\begin{theorem}[Spectral mapping theorem]\label{thm:SMT}
    Let $\sigma_{\bullet }\in \left\{\sigma ,\ \sigma_{\operatorname{ess}} ,\ \sigma_{\operatorname{p}} ,\ \sigma_{\operatorname{pp}} ,\ \sigma_{\operatorname{ac}} ,\ \sigma_{\operatorname{sc}}\right\}$ denote the spectrum, essential, point, pure point, absolutely continuous or singular continuous spectrum, respectively. Then
    \begin{equation}\label{eq:SMT}
        \sigma_{\bullet }\left( W^{(n)}\right) =\Sigma _{\bullet }^{\ker } \cup \Sigma _{\bullet }^{( 0,1)} ,
    \end{equation}
    where $\Sigma _{\bullet }^{\ker }$ and $\Sigma _{\bullet }^{( 0,1)}$ are given as follows.
    \begin{enumerate}
        \item[(a)] \emph{Kernel sectors.} The kernel contribution to $\sigma_{\bullet}\left(W^{(n)}\right)$ is
        \[
            \Sigma _{\bullet }^{\ker } =\begin{cases}
                \left\{ \omega _{T} :T\in \{ A,A^{*} ,B,B^{*}\} ,\ \ker T\neq \{ 0\}\right\} & ( \sigma _{\bullet } =\sigma ,\sigma _{\operatorname{p}} ,\sigma _{\operatorname{pp}}) ,\\
                \left\{ \omega _{T} :T\in \{ A,A^{*} ,B,B^{*}\} ,\ \dim \ker T=\infty \right\} & ( \sigma _{\bullet } =\sigma _{\operatorname{ess}}) ,\\
                \emptyset & ( \sigma _{\bullet } =\sigma _{\operatorname{ac}} ,\sigma _{\operatorname{sc}}) ,
            \end{cases}
        \]
        where
        $\omega_A=\prod_{k=1}^{n}e^{i(\theta_k-\phi_k)}$,
        $\omega_B=\prod_{k=1}^{n}e^{-i(\theta_k+\phi_k)}$,
        $\omega_{A^*}=\overline{\omega_A}$ and
        $\omega_{B^*}=\overline{\omega_B}$.
        The corresponding eigenstates lie in the following sectors,
        whose dimensions give their multiplicities:
       \[
\begin{aligned}
    \omega_A &: \ker A
    \quad(\dim\ker A),
    &\qquad
    \omega_{A^*} &: \ker D
    \quad(\dim\ker A^*),\\
    \omega_B &: \ker B
    \quad(\dim\ker B),
    &\qquad
    \omega_{B^*} &: \ker C
    \quad(\dim\ker B^*).
\end{aligned}
\]
        Here $\ker D=B^*\ker A^*$ and $\ker C=A^*\ker B^*$.
        \item[(b)] \emph{Qubitized sector.} The restriction of $W^{(n)}$ to $\mathcal L_{( 0,1)}$ has spectrum
        \begin{equation}\label{eq:SMT-band}
            \Sigma _{\bullet }^{( 0,1)} :=\sigma_{\bullet }\left( W^{(n)}|_{\mathcal L_{(0,1)}}\right)
            =
            \left\{
                z\in\mathbb T:
                \operatorname{Re} z
                \in
                \operatorname{Re}\Pi_n\left( \sigma_{\bullet }( AA^{*}|_{(0,1)})\right)
            \right\} ,
        \end{equation}
        where for $\sigma_{\operatorname{p}}$, $\sigma_{\operatorname{pp}}$, $\sigma_{\operatorname{ac}}$ and $\sigma_{\operatorname{sc}}$ we assume that $\operatorname{Re}\Pi_n$ is non-constant; if $\operatorname{Re}\Pi_n\equiv c$ is constant and $\mathcal L_{( 0,1)} \neq \{ 0\}$, then $\Sigma _{\operatorname{p}}^{( 0,1)} =\Sigma _{\operatorname{pp}}^{( 0,1)} =\{ z\in \mathbb{T} :\operatorname{Re} z=c\}$ and $\Sigma _{\operatorname{ac}}^{( 0,1)} =\Sigma _{\operatorname{sc}}^{( 0,1)} =\emptyset$.
    \end{enumerate}
\end{theorem}

\begin{proof}
    See \cref{app:smt}.
\end{proof}

Note that since $AA^{*}|_{(0,1)}$ is the restriction of $AA^{*}$ to the spectral subspace of $( 0,1)$,
\[
    \sigma ( AA^{*}|_{(0,1)}) =\overline{\sigma \left( AA^{*}\right) \cap ( 0,1)} .
\]
In particular, the endpoints $0$ and $1$ belong to $\sigma ( AA^{*}|_{(0,1)})$ whenever they are accumulation points of $\sigma \left( AA^{*}\right) \cap ( 0,1)$.

As a tool for spectral analysis, \cref{thm:SMT} plays a role analogous to the Fourier transform in translation-invariant quantum systems: rather than diagonalising the full unitary $W^{(n)}$ on $\mathcal{H}$ directly, it reduces the spectral problem to that of the single self-adjoint operator $AA^{*}$ on the smaller subspace $\mathfrak{K}$ together with the four kernel sectors, without requiring translation symmetry. 
\subsection{Spectral design\label{sec:design}}

We now turn \cref{thm:SMT} into a tool for engineering the spectrum of $W^{(n)}$. Whereas \cref{thm:SMT} determines the spectrum from a given phase sequence through $\operatorname{Re}\Pi_{n}$, in applications one prescribes a target polynomial $g=\operatorname{Re}\Pi_{n}\colon[0,1]\to[-1,1]$, mapping the spectral interval $[0,1]$ of $AA^{*}$ to the real parts $[-1,1]$ on $\mathbb{T}$, to reshape spectral bands and position kernel eigenvalues. Combining \cref{thm:SMT} with QSVT achievability \cite{QSVT} characterises which spectral transformations are realisable, which are ruled out, and how many steps $n$ are required to stretch a band, separate two bands, or isolate and embed kernel eigenvalues. To measure arc lengths on $\mathbb{T}$, for any closed interval $J=[u_{-},u_{+}]\subseteq[-1,1]$ we write
\begin{equation}\label{eq:arclength}
    \ell(J):=\arccos u_{-}-\arccos u_{+}\in[0,\pi]
\end{equation}
for the arc length of each of the two conjugate arcs $\{z\in\mathbb{T}:\operatorname{Re} z\in J\}$. By \eqref{eq:PQ}, a single step ($n=1$) gives the affine polynomial $\operatorname{Re}\Pi_{1}(x)=\cos(\theta_{1}-\phi_{1})-2\sin\theta_{1}\sin\phi_{1}\,x$, which stretches the input domain $[0,1]$ at most onto $[-1,1]$ via $x\mapsto 2x-1$ (attained at $\theta_{1}=3\pi/2$, $\phi_{1}=\pi/2$). Under this maximal single-step stretch, an input interval $I=[x_{-},x_{+}]\subseteq[0,1]$ maps to $2I-1:=[2x_{-}-1,2x_{+}-1]\subseteq[-1,1]$ with arc length $\tilde{\ell}(I):=\ell(2I-1)$ on $\mathbb{T}$; an $n$-step polynomial $g$ can then stretch this single-step arc length by at most a factor of $n$.

\begin{theorem}[Spectral design]\label{thm:design}
    Let $\sigma _{\bullet }$ be any of the six spectra of \cref{thm:SMT}.
    \begin{enumerate}
        \item[(a)] \emph{Achievability.} A polynomial $g\in\mathbb{R}[X]$ is realised as $\operatorname{Re}\Pi_{n}=g$ for some phases $\Phi\in[0,2\pi)^{2n}$ if and only if $\deg g\leqslant n$ and $|g|\leqslant 1$ on $[0,1]$ (called \emph{admissible of degree $n$}). By \cref{thm:SMT} (with $g$ non-constant for $\sigma_{\operatorname{p}},\sigma_{\operatorname{pp}},\sigma_{\operatorname{ac}},\sigma_{\operatorname{sc}}$),
        \begin{equation}\label{eq:design}
            \sigma _{\bullet }\left( W^{( n)}\right)
            =\Sigma _{\bullet }^{\ker } \cup \Sigma _{\bullet }^{( 0,1)} ,\qquad
            \Sigma _{\bullet }^{( 0,1)} =\left\{ z\in \mathbb{T} :\operatorname{Re} z\in g\bigl( \sigma _{\bullet }( AA^{*}|_{(0,1)})\bigr)\right\} .
        \end{equation}
        Here, $\Sigma _{\bullet }^{( 0,1)}$ is symmetric under $z\mapsto \overline{z}$ and $g$ maps each interval of $\sigma_{\bullet}(AA^{*}|_{(0,1)})$ to a single interval (conjugate arcs are never separated and no gap can be opened inside a band).
        \item[(b)] \emph{Degree cost.} For bands $I,I_{1},I_{2}\subseteq\sigma_{\operatorname{ac}}(AA^{*}|_{(0,1)})$, let $n_{\min}$ denote the smallest admissible degree:
        \begin{enumerate}
            \item[(i)] \emph{Reshaping a band.} To achieve $g(I)=J=[m,M]\subseteq[-1,1]$,
            \begin{equation}\label{eq:nmin}
                \frac{\ell(J)}{\tilde{\ell}(I)} \ \leqslant\ n_{\min } \ \leqslant\ \left\lceil \frac{2\pi }{\tilde{\ell}(I)}\right\rceil .
            \end{equation}
            The upper bound follows from the construction $g( x) =\frac{M-m}{2} T_{\lceil 2\pi /\tilde{\ell}(I) \rceil}( 2x-1) +\frac{M+m}{2}$.
            \item[(ii)] \emph{Separating two bands.} To separate $g(I_{1})$ and $g(I_{2})$ across the gap $G=[x_{1},x_{2}]$ ($x_{1}=\max I_{1}<\min I_{2}=x_{2}$) by a gap $J\subseteq[-1,1]$ of given arc length $\ell(J)\in(0,\pi)$,
            \begin{equation}\label{eq:band-gap}
                \frac{\ell(J)}{\tilde{\ell}(G)} \ \leqslant\ n_{\min } \ \leqslant\ O\!\left( \frac{1}{x_{2}-x_{1}}\log \frac{2\pi}{\pi -\ell(J) }\right) ,
            \end{equation}
            with the upper bound from the sign-function polynomial of \cite{QSVT}.
            \item[(iii)] \emph{Separating the kernel eigenvalue from a band.} For $I=[x_{-},x_{+}]$ ($x_{-}>0$), $G=[0,x_{-}]$, and $n\bigl( 2\tilde{\ell}(G)+\tilde{\ell}(I) \bigr) \leqslant 2\pi$, the largest gap $J\subseteq[-1,1]$ separating $g(0)$ and $g(I)$ at degree $n$ has arc length
            \begin{equation}\label{eq:kernel-gap}
                \max_{g}\ell(J) =n\,\tilde{\ell}(G) ,
            \end{equation}
            attained by $g( x) =T_{n}( 1-2x)$ (and at $x=1$ with $x\mapsto 1-x$).
            \item[(iv)] \emph{Merging spectra.} For any distinct $x_{0},x_{1}\in[0,1]$, achieving $g(x_{0})=g(x_{1})$ with non-constant $g$ requires $n_{\min } =2$, attained by $g( x) =2\bigl(\tfrac{x-c}{r}\bigr)^{2} -1$ with $c=(x_{0}+x_{1})/2$ and $r=\max(c,1-c)$. For $x_{0}=0$ and $x_{1}\in(x_{-},x_{+})\subseteq I$ ($x_{-}>0$), $g(0)=g(x_{1})\in(-1,1)$ and $g'(x_{1})\neq 0$, embedding $\omega _{A} ,\omega _{A^{*}}$ into the interior of the absolutely continuous arcs of $I$ as a bound state in the continuum whenever $\ker A$ or $\ker A^{*}$ is non-zero.
        \end{enumerate}
    \end{enumerate}
\end{theorem}

\begin{proof}
    See \cref{app:design}.
\end{proof}

In all cases $g$ acts on all bands at once, so prescribing the images of several bands independently requires correspondingly more degrees of freedom.

\begin{example}[Homogeneous phase sequence]\label{ex:homo}
    A simple class of spectral transformations is obtained by repeating a single QSVT step. Set $\theta_k=\theta$ and $\phi_k=\phi$ for all $k$, so that $W^{(n)}=W_1^n$. For $x\in[0,1]$, put
    \[
        \gamma(x)=\cos(\theta-\phi)-2\sin\theta\sin\phi\,x,
        \qquad
        \zeta(x)=\sin\theta\cos\phi+\cos\theta\sin\phi\,(2x-1).
    \]
    The corresponding polynomials have the Chebyshev representations
    \begin{align*}
        \Pi_n(x)&=T_n(\gamma(x))+i\zeta(x)U_{n-1}(\gamma(x)),\\
        \Phi_n(x)&=2e^{i\theta}\sin\phi\,U_{n-1}(\gamma(x)),
    \end{align*}
    where $T_n$ and $U_{n-1}$ are the Chebyshev polynomials of the first and second kind.
    Indeed, by \eqref{eq:PQ} and \eqref{eq:qsp-fibre}, $\Pi_1(x)=\gamma(x)+i\zeta(x)$ and $\Phi_1(x)=2e^{i\theta}\sin\phi$, so the fibre matrix $M_1(x)$ has trace $2\gamma(x)$ and determinant $1$ by \eqref{eq:unitarity}. The Cayley--Hamilton identity $M_1(x)^2=2\gamma(x)M_1(x)-I$ yields
    \[
        M_n(x)=M_1(x)^n
        =U_{n-1}(\gamma(x))M_1(x)-U_{n-2}(\gamma(x))I
        \qquad(n\geqslant1),
    \]
    with $U_{-1}=0$. Comparing the first-row entries and using $\gamma U_{n-1}(\gamma)-U_{n-2}(\gamma)=T_n(\gamma)$ gives the formulas above, with the endpoint values obtained by continuity. No diagonalisation is used, so the formulas also hold at $\gamma(x)=\pm1$. In particular, since $\zeta(x)$ is real,
    \[
        g(x)=\operatorname{Re}\Pi_n(x)=T_n(\gamma(x)).
    \]
    At $\theta=\phi=\pi/2$, this reduces to $g(x)=T_n(1-2x)$.
    For $n=1$, the relation $\operatorname{Re}z=1-2x$ recovers the spectral mapping formula for a product of two involutions, including the overall minus sign in the QSVT step \cite{SMT}. For general $n$, this gives the Chebyshev polynomial
    used in \cref{thm:design}(b)(iii).
\end{example}

\subsection{Symmetry and topological invariants\label{sec:symmetry}}

On the qubitized sector $\mathcal L_{(0,1)}$, the $SU(2)$ structure of the fibre matrix $M_n(x)$ in \cref{prop:L-decomp} enforces a Kramers-type conjugation symmetry. 
By \cref{prop:L-decomp}, $W^{(n)}|_{\mathcal L_{(0,1)}}=J\,M_n(AA^*|_{(0,1)})\,J^*$. By \eqref{eq:qsp-fibre} and \eqref{eq:unitarity}, $M_n(x)\in SU(2)$ for $x\in[0,1]$. By the spectral theorem, fix a unitary identification $\mathfrak K_{(0,1)}\simeq\bigoplus_\alpha L^2((0,1),\mu_\alpha)$ under which $AA^*|_{(0,1)}$ acts as multiplication by the variable, $\bigl((AA^*|_{(0,1)})f\bigr)_\alpha(x)=x\,f_\alpha(x)$ for every $f=(f_\alpha)_\alpha$. Under this identification, $(\mathfrak j f)_\alpha(x):=\overline{f_\alpha(x)}$ defines an antiunitary $\mathfrak j$ with $\mathfrak j^2=I$ commuting with $AA^*|_{(0,1)}$, so that $\mathfrak j\,h(AA^*|_{(0,1)})\,\mathfrak j=\overline h(AA^*|_{(0,1)})$ for every continuous function $h$ on $[0,1]$. Since every $M\in SU(2)$ satisfies $\varepsilon\overline M\varepsilon^{-1}=M$ with $\varepsilon=\bigl[\begin{smallmatrix}0&-1\\1&0\end{smallmatrix}\bigr]$, these identities show that the antiunitary
\[
    \Xi:=J\begin{bmatrix}0&-\mathfrak j\\ \mathfrak j&0\end{bmatrix}J^*
\]
satisfies
\begin{equation}\label{eq:Xi-sym}
    \Xi^2=-I,\qquad
    \Xi W^{(n)}\Xi^{-1}=W^{(n)}\quad\text{on }\mathcal L_{(0,1)}.
\end{equation}
Consequently, for every $z\in\mathbb T$,
\[
    \Xi\ker\!\left(W^{(n)}|_{\mathcal L_{(0,1)}}-z\right)
    =\ker\!\left(W^{(n)}|_{\mathcal L_{(0,1)}}-\overline z\right),
\]
so conjugate eigenvalues have equal multiplicities. Since $\Xi^2=-I$, every $v\in\mathcal L_{(0,1)}$ satisfies $v\perp\Xi v$, so each eigenvalue $z=\pm1$ on $\mathcal L_{(0,1)}$ has even or infinite multiplicity.

The kernel sectors need not have this conjugation symmetry. By \cref{thm:SMT}(a), the $A$ and $B$ kernel sectors correspond to the conjugate eigenvalue pairs $\{\omega_A,\omega_{A^*}\}$ and $\{\omega_B,\omega_{B^*}\}$, respectively, with sector dimensions $\dim\ker A$, $\dim\ker A^*$ and $\dim\ker B$, $\dim\ker B^*$. These dimensions are independent of the QSP phase sequence,
but $\dim\ker A$ need not equal $\dim\ker A^*$,
and likewise for $B$. The Fredholm index measures this imbalance if $A$, respectively $B$, is Fredholm;
an operator is Fredholm if its kernel and adjoint kernel
are finite dimensional and its range is closed.
The following lemma expresses this condition
as a spectral gap and gives an equivalent criterion
in terms of the essential spectrum.

\begin{lemma}[Fredholmness \cite{CGSVW18,CGWW21}]\label{lem:fredholm-endpoints}
$A$ is Fredholm precisely when $\ker A$ and $\ker A^*$
are finite dimensional and $\sigma(AA^*|_{(0,1)})$
stays a positive distance from zero.
Equivalently,
    \[
    \begin{aligned}
        A\text{ is Fredholm}
        &\iff
        \begin{cases}
            \dim\ker A,\ \dim\ker A^*<\infty,\\
            0\notin\sigma\!\left(AA^*|_{(0,1)}\right)
        \end{cases}\\
        &\iff 0\notin\sigma_{\operatorname{ess}}(A^*A)
            \cup\sigma_{\operatorname{ess}}(AA^*).
    \end{aligned}
    \]
    \noindent Similarly for $B$,
    \[
    \begin{aligned}
        B\text{ is Fredholm}
        &\iff
        \begin{cases}
            \dim\ker B,\ \dim\ker B^*<\infty,\\
            1\notin\sigma\!\left(AA^*|_{(0,1)}\right)
        \end{cases}\\
        &\iff 0\notin\sigma_{\operatorname{ess}}(B^*B)
            \cup\sigma_{\operatorname{ess}}(BB^*).
    \end{aligned}
    \]
\end{lemma}
\begin{proof}
    See \cref{app:fredholm}.
\end{proof}
This is a standard characterisation of Fredholm operators; see also \cite{CGSVW18,CGWW21}.
The following theorem shows that non-zero Fredholm indices
guarantee eigenstates whose existence is protected under
specified perturbations of the QSVT encoding.
\begin{theorem}[Topological invariant]\label{thm:index}
    Fix the projections $\Pi_{\mathfrak H},\Pi_{\mathfrak K}$ and a QSP phase sequence of length $2n$, and let $\omega_A,\omega_{A^*},\omega_B,\omega_{B^*}$ be as in \cref{thm:SMT}.
    \begin{enumerate}
        \item[(a)] \emph{Indices and protected states.}
        If $A$, respectively $B$, is Fredholm, define
        \[
            \ind A=\dim\ker A-\dim\ker A^*,\qquad
            \ind B=\dim\ker B-\dim\ker B^*.
        \]
        Under each listed condition, $W^{(n)}$ has at least
the indicated number of linearly independent eigenstates
in the given sector at the specified eigenvalue:
        \[
        \begin{array}{c|c|c|c}
            \textnormal{Condition}&\textnormal{Protected states}&\textnormal{Sector}&\textnormal{Eigenvalue}\\ \hline
            \ind A>0&|\ind A|&\ker A&\omega_A\\
            \ind A<0&|\ind A|&\ker D&\omega_{A^*}\\
            \ind B>0&|\ind B|&\ker B&\omega_B\\
            \ind B<0&|\ind B|&\ker C&\omega_{B^*}
        \end{array}
        \]
        Under the perturbations in (b), these bounds persist,
        although the eigenvectors may change.
        \item[(b)] \emph{Topological protection.}
        \begin{enumerate}
            \item[(i)] \emph{Compact perturbations.}
            Let $U'$ be unitary,
            $\Pi'_{\mathfrak H},\Pi'_{\mathfrak K}$ be orthogonal projections,
            and $(\theta'_k,\phi'_k)_{k=1}^{n'}$, $n'\geqslant1$,
            be any new phase sequence.
            Suppose that $U'-U$, $\Pi'_{\mathfrak H}-\Pi_{\mathfrak H}$
            and $\Pi'_{\mathfrak K}-\Pi_{\mathfrak K}$ are compact,
            with $\|\Pi'_{\mathfrak H}-\Pi_{\mathfrak H}\|<1$ and
            $\|\Pi'_{\mathfrak K}-\Pi_{\mathfrak K}\|<1$.
            Then the corresponding blocks
            $A'$ and $B'$ (defined as in \eqref{eq:blocks})
            remain Fredholm whenever the respective original blocks are,
            with their indices preserved:
            \[
                \ind A'=\ind A,\qquad \ind B'=\ind B.
            \]
            The guarantees in (a) therefore hold for the primed encoding,
            with
            $\omega'_A=\prod_{k=1}^{n'}e^{i(\theta'_k-\phi'_k)}$ and
            $\omega'_B=\prod_{k=1}^{n'}e^{-i(\theta'_k+\phi'_k)}$
            replacing $\omega_A$ and $\omega_B$, respectively.

            \item[(ii)] \emph{Continuous deformations.}
            Under norm-continuous deformations of
            $U,\Pi_{\mathfrak H},\Pi_{\mathfrak K}$,
            $\ind A$, respectively $\ind B$, is constant
            as long as $A$, respectively $B$, remains Fredholm.
            In particular, sufficiently small norm perturbations
            of the encoding preserve each Fredholm index.
        \end{enumerate}
    \end{enumerate}
\end{theorem}

\begin{proof}
    See \cref{app:index}.
\end{proof}

Under the compact-perturbation assumptions of
\cref{thm:index}(b)(i), let ${W'}^{(n)}$ be the QSVT operator
constructed from $U',\Pi'_{\mathfrak H},\Pi'_{\mathfrak K}$
with the same phase sequence $\Phi$ as $W^{(n)}$.
Then ${W'}^{(n)}-W^{(n)}$ is compact, so
\[
    \sigma_{\operatorname{ess}}({W'}^{(n)})
    =\sigma_{\operatorname{ess}}(W^{(n)}).
\]
\Cref{thm:index} separates spectral control from protection:
the QSP phases move the bands and kernel eigenvalues,
while the indices guarantee eigenstates under the allowed
perturbations of the encoding.
At fixed phases, the guaranteed eigenvalues remain unchanged,
although the eigenvectors may change.
The protection does not require these eigenvalues to be isolated:
they may lie inside the essential spectrum, which is what
makes the bound states in the continuum of \cref{sec:example} possible.
With $\Pi_{\mathfrak H},\Pi_{\mathfrak K}$ fixed,
encodings with different indices cannot be connected by
a norm-continuous unitary path along which the corresponding
operator $A$ or $B$ remains Fredholm.

In finite-dimensional QSVT, the indices are fixed by
dimension differences and are independent of $U$.
Extending QSVT to infinite dimensions allows these same
indices to become non-trivial topological invariants
of the encoding unitary $U$, even with the projections fixed.
Indeed, $\ind A$ is the Avron--Seiler--Simon index of the pair of projections
$U\Pi_{\mathfrak H}U^*$ and $\Pi_{\mathfrak K}$ \cite{ASS}:
\begin{equation*}
    \ind A
    =\dim(U\mathfrak H\cap\mathfrak K^\perp)
     -\dim(\mathfrak K\cap U\mathfrak H^\perp).
\end{equation*}
This index has a direct physical interpretation:
in quantum Hall systems, it relates the change in occupied
states under flux insertion to quantised Hall transport
\cite{ASSHall,ASSCharge}.
Its stability explains how an integer-valued physical
response can survive perturbations.
Here the same index structure appears in the QSVT encoding:
a non-zero index forces eigenstates to remain present,
while the QSP phases control their eigenvalues. For chiral symmetric $U$, taking $\mathfrak H$ and $\mathfrak K$
as the positive and negative chiral sectors identifies
$\ind A$, when $A$ is Fredholm, with the Witten index
\cite{SuzukiSUSY,MatsuzawaWitten,SuzukiTanakaWitten}.
For chiral symmetric walks, \cite{CGWW21} expresses the symmetry indices at $\pm1$ through the Fredholm indices of the blocks of a half-step unitary. With $U$ as the half-step, these relations give the eigenstate bounds of \cref{thm:index}(a) in the special case $n=1$, $\theta_1=\phi_1=\pi/2$ and $\Pi_{\mathfrak H}=\Pi_{\mathfrak K}$, up to an overall sign (\cref{sec:example}). \Cref{thm:index} does not require the symmetry: for generic phases $W^{(n)}$ is not chiral symmetric, the same indices persist, and the protected eigenvalues $\omega_A,\omega_B$ can be moved away from $\pm1$.

Combining \cref{thm:design,thm:index} allows spectral design
while protecting the existence of eigenstates.
For Fredholm $A$ with $\ind A\neq0$, the index guarantees
at least $|\ind A|$ kernel eigenstates at $\omega_A$ or
$\omega_{A^*}$, independently of the QSP phases.
The phases control their position relative to the bands:
they can open a gap under the conditions of
\cref{thm:design}(b)(iii), or place a kernel eigenvalue
inside an absolutely continuous band in two steps
by \cref{thm:design}(b)(iv), provided such an input band exists.
The latter produces a bound state in the continuum (BIC)
\cite{BICreview}: a normalisable eigenstate whose quasi-energy
lies inside an absolutely continuous band.
Crucially, Fredholmness requires finite-dimensional kernel
sectors and a gap separating the non-zero singular values
of $A$ from zero, but no quasi-energy gap around the protected
eigenvalues of $W^{(n)}$.
Thus the index protects the existence of these states
even inside a continuous band.
Under the perturbations allowed by \cref{thm:index},
they remain BICs provided an absolutely continuous band
continues to contain their eigenvalues.
A sufficient condition is available for trace-class perturbations. Under the compact-perturbation assumptions of \cref{thm:index}(b)(i), fix the QSP phase sequence and suppose, in addition, that $A'{A'}^*-AA^*$ is trace class after the unitary identification of $\mathfrak K'$ with $\mathfrak K$ used in \cref{app:index}. The Kato--Rosenblum theorem \cite{KatoTraceClass} preserves the absolutely continuous parts up to unitary equivalence, so \cref{thm:SMT}(b) preserves the absolutely continuous bands of the QSVT evolution. The index and the protected eigenvalues are also unchanged; hence an index-guaranteed BIC remains a BIC. This applies, in particular, when the identified encoding blocks change by finite-rank operators.
The next section gives a concrete construction.

\section{Quantum walks}\label{sec:example}

Quantum walks are natural test cases for our framework. They play a central role both in quantum algorithms and in condensed matter physics as models of topological phases, two directions that have largely been studied separately. We present two of them as QSVT operators. Szegedy's walk is the expected case: the embedded operator is the discriminant of the underlying Markov chain, and the algorithmic content lives in the qubitized sector $\mathcal L_{(0,1)}$, where \cref{thm:poly} applies the polynomial to the discriminant. Here the discarded sectors carry information: on a finite graph their indices are fixed by the Euler characteristic of the graph. The split-step walk is the complementary case. It fits the same framework, but what matters for it is the part that QSVT normally discards: the kernel sectors. Their dimensions give its two topological numbers as $(\ind A,\ind B)$ of \cref{thm:index}, and \cref{thm:design} then controls the eigenvalues of the resulting edge states. 

At $\theta_k=\phi_k=\pi/2$, the constructions for Szegedy's
and split-step walks share a reflection structure:
with $\Gamma_{\mathfrak H}=2\Pi_{\mathfrak H}-I$ and
$\Gamma_{\mathfrak K}=2\Pi_{\mathfrak K}-I$,
$W_k=(-\Gamma_{\mathfrak H})(U^*\Gamma_{\mathfrak K}U)$
is a product of two self-adjoint involutions, and $W^{(n)}$ has chiral symmetry
$\Gamma_{\mathfrak H}W^{(n)}\Gamma_{\mathfrak H}=(W^{(n)})^*$.
Up to the sign and a unitary identification of $\Gamma_{\mathfrak K}$ with $\Gamma_{\mathfrak H}$, $W_1$ is the half-step form of \cite{CGWW21} with the identified encoding unitary as the half-step unitary, and $\ind A,\ind B$ are among their classifying indices. Conversely, every chiral symmetric unitary with an essential gap at $\pm1$ is, up to the sign, of this form, so chiral symmetric walks \cite{SMT,CGSVWW,CGGSVWW,CGSVW18} are the $\pi/2$ slice of operator QSVT.
Our results extend both constructions to arbitrary phases
without assuming chiral symmetry:
the antiunitary symmetry \eqref{eq:Xi-sym} on
$\mathcal L_{(0,1)}$ and the indices survive
(\cref{thm:index}), while \cref{thm:design} controls the spectrum.

\subsection{Szegedy's quantum walk\label{sec:szegedy}}

Szegedy's quantum walk construction \cite{szegedy} provided a building block for a broad class of quantum search algorithms \cite{MNRS}. Its use of reflections and spectral relations between a classical Markov chain and a unitary evolution anticipated ideas later generalised by qubitization and QSVT \cite{LC19,QSVT}. Here we revisit this established construction in our QSVT notation, using the arc-based formulation of \cite{SMT,crystal}, to make the encoding unitary, projections and kernel sectors concrete and intuitive as a walk on the graph (\cref{fig:szegedy}).

Let $G=( V,E)$ be a (possibly infinite) graph, let $\mathcal{A} =\{( a,b) \mid \{a,b\} \in E\}$ be its set of directed arcs, and let $\mathcal{H} =\ell ^{2}(\mathcal{A})$. We write $o( e)$ for the origin of an arc, $t( e)$ for its terminus and $\overline{e}$ for its reversal. Given weights $p_{e} \geqslant 0$ with $\sum _{o( e) =x} p_{e} =1$ for every $x$, put
\begin{equation}\label{eq:szegedy-data}
  \ket{\psi ( x)} :=\sum _{o( e) =x}\sqrt{p_{e}}\ket{e} ,\qquad
  \Pi _{\mathfrak{H}} =\Pi _{\mathfrak{K}} :=\sum _{x\in V}\ket{\psi ( x)}\bra{\psi ( x)} ,
\end{equation}
and let the coin and the shift be
\begin{equation*}
  \mathcal{C} :=2\Pi _{\mathfrak{H}} -I,\qquad
  S:=\sum _{e\in \mathcal{A}}\ket{\overline{e}}\bra{e} ,
\end{equation*}
with $S=S^{*} =S^{-1}$. Their actions on the graph are illustrated in \cref{fig:szegedy}.
\begin{figure}[!ht]
  \centering
  \includegraphics[width=0.95\textwidth]{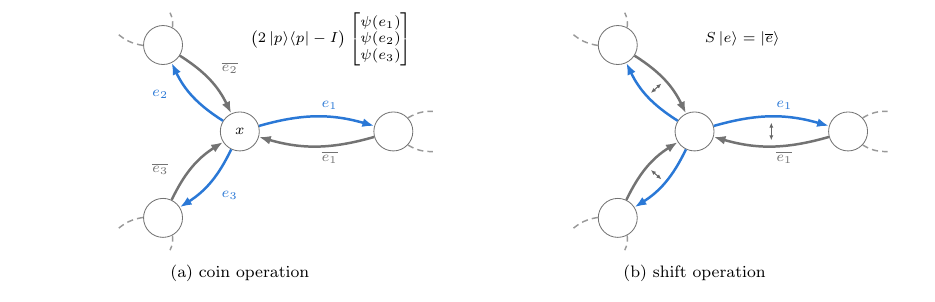}
  \caption{Szegedy's walk. (a) The coin acts locally on the blue
arcs at each vertex $x$ by reflecting about
$\ket{p}=[\sqrt{p_{e_1}},\sqrt{p_{e_2}},\sqrt{p_{e_3}}]^T$,
whose entries are square roots of the Markov chain's transition
probabilities at $x$, with $\psi(e)$ denoting
the coefficient of $\ket{e}$.
(b) The shift swaps amplitudes on each arc and its reverse.}
  \label{fig:szegedy}
\end{figure}

Taking $U=S$, $\mathfrak{H} =\mathfrak{K} =\Ran \Pi _{\mathfrak{H}}$ and $\theta _{k} =\phi _{k} =\pi /2$ gives $R_{\mathfrak{H}}( \theta _{k}) =R_{\mathfrak{K}}( \phi _{k}) =i\mathcal{C}$, so that
\begin{equation*}
  W_{k} =( i\mathcal{C}) S( i\mathcal{C}) S=-(\mathcal{C} S)^{2} .
\end{equation*}
Thus $-W_k$ is unitarily equivalent to Szegedy's original two-reflection walk \cite{szegedy}. The spectrum of the associated one-step arc walk $\mathcal CS$ is given by the spectral mapping theorem of \cite{SMT}. The embedded operator is
\begin{equation*}
  A=\Pi _{\mathfrak{H}} S\Pi _{\mathfrak{H}} =\sum _{( x,y) \in \mathcal{A}}\sqrt{p_{( x,y)} p_{( y,x)}}\ket{\psi ( x)}\bra{\psi ( y)} ,
\end{equation*}
which, under the identification $\ket{\psi ( x)} \mapsto \ket{x}$, is the discriminant
\begin{equation*}
  A\simeq \sqrt{P\circ P^{T}} ,\qquad P=\sum _{( x,y) \in \mathcal{A}} p_{( x,y)}\ket{x}\bra{y} ,
\end{equation*}
of the Markov chain with transition matrix $P$, the Hadamard product and square root being taken entrywise. For a reversible chain with positive reversible
measure $\pi$, detailed balance gives
$A=\Lambda^{1/2}P\Lambda^{-1/2}$, where
$\Lambda=\operatorname{diag}(\pi)$ and
$\Lambda^{1/2}\colon\ell^2(V,\pi)\to\ell^2(V)$ is unitary. $AA^*=A^2$, and \cref{thm:SMT} determines the spectrum
of $W^{(n)}$ from the squared spectrum of the random walk $P$.

Allowing the phases to depend on the step gives the time-dependent walk
\begin{equation*}
  W_{k} =\mathcal{C}_{\theta _{k}} S\,\mathcal{C}_{\phi _{k}} S,\qquad
  \mathcal{C}_{\theta } :=e^{i\theta } \Pi _{\mathfrak{H}} +e^{-i\theta }\left( I-\Pi _{\mathfrak{H}}\right) .
\end{equation*}
The spectrum of $W^{(n)}=W_n\cdots W_1$ is described by \cref{thm:SMT}, and \cref{thm:design} allows it to be controlled through the phase sequence. The algorithmic importance of spectral transformations is well established: quantum-walk search uses spectral discrimination to implement approximate reflections \cite{MNRS}, and quantum fast-forwarding uses polynomial approximations to powers of the discriminant to accelerate reversible Markov-chain dynamics \cite{QFF}. 

The QSP phases move the kernel eigenvalues
$\omega_A,\omega_{A^*},\omega_B,\omega_{B^*}$
without changing the kernel sectors or their indices
(\cref{thm:SMT}(a), \cref{thm:index}(a)).
On a finite graph $G$,
$\dim\mathfrak H=\dim\mathfrak K=|V|$ and
$\dim\mathfrak H^\perp=2|E|-|V|$, so rank--nullity gives
\begin{equation}\label{eq:szegedy-indB}
  \ind A=0,\qquad
  \ind B=\dim\mathfrak H^\perp-\dim\mathfrak K
  =2|E|-2|V|=-2\chi(G),
\end{equation}
where $\chi(G)=|V|-|E|=b_0-b_1$ is the Euler characteristic, $b_0$ and $b_1=|E|-|V|+b_0$ being the Betti numbers of $G$ (counting connected components and independent cycles); thus $\ind B=2(b_1-b_0)$. The index of \cref{thm:index} is a topological invariant of the graph, carried by the sectors that QSVT discards. For uniform weights $p_e=1/\deg(o(e))$, the simple random walk whose quantum counterpart is the Grover walk \cite{crystal,kubota19}, the individual sectors are graph-theoretic as well. Write $c_{\operatorname{bip}}$ for the number of bipartite components and $\nu(G)=\dim\ker\operatorname{Adj}(G)$ for the adjacency nullity:
\begin{equation}\label{eq:szegedy-sectors}
  \dim \ker A=\dim \ker A^{*} =\nu ( G) ,\quad
  \dim \ker B^{*} =b_{0} +c_{\operatorname{bip}} ,\quad
  \dim \ker B=2b_{1} -b_{0} +c_{\operatorname{bip}} .
\end{equation}
Indeed, under $\ket{\psi(x)}\mapsto\ket{x}$, the self-adjoint operator $A$ is $D^{-1/2}\operatorname{Adj}(G)D^{-1/2}$, where $D$ is the diagonal degree matrix. Thus $\ker A$ and $\ker A^*$ both have dimension $\nu(G)$ \cite{BrouwerHaemers}. Also, $\ker B^*=\ker(A^2-I)$: the eigenvalue $+1$ occurs once per connected component, and $-1$ once per bipartite component \cite{crystal}. This gives $\dim\ker B^*=b_0+c_{\operatorname{bip}}$, and $\dim\ker B=\ind B+\dim\ker B^*$ gives the last dimension in \eqref{eq:szegedy-sectors}. Note the contrast: $\nu(G)$ and $c_{\operatorname{bip}}$ are not topological, since both change when an edge is subdivided, whereas $\ind B$ does not. 

\subsection{Split-step quantum walk\label{sec:ss}}
The split-step walk \cite{Kitagawa,kitagawa2}, in the slightly simplified form used here, is the standard model for one-dimensional topological phases. It acts on $\mathcal{H} =\mathcal{H}_{P} \otimes \mathcal{H}_{C}$, where $\mathcal{H}_{C}$ is the two-dimensional coin space spanned by $\left\{\ket{\uparrow } ,\ket{\downarrow }\right\}$ and $\mathcal{H}_{P} =\ell ^{2}( \mathbb{Z})$ is the position space. Its time evolution is
\begin{equation*}
    U_{ss} =R( \vartheta _{2}) S_{\uparrow } R( \vartheta _{1}) S_{\downarrow } ,
\end{equation*}
where the coin $R$ rotates the spin about the $y$-axis through a position-dependent angle $\vartheta :\mathbb{Z}\rightarrow [ -\pi ,\pi ]$. Here $\vartheta_1,\vartheta_2$ are the coin angles of the model, functions of position and fixed in time; they are distinct from the QSVT phases $\theta_k,\phi_k$, which depend on the step and not on position. We have
\begin{equation*}
    R(\vartheta )=\sum _{x} \ket{x}\bra{x}\otimes \begin{bmatrix}
        \cos\left(\frac{\vartheta ( x)}{2}\right) & -\sin\left(\frac{\vartheta ( x)}{2}\right) \\
        \sin\left(\frac{\vartheta ( x)}{2}\right) & \cos\left(\frac{\vartheta ( x)}{2}\right)
    \end{bmatrix} .
\end{equation*}
Here $S_{\uparrow }$ shifts the up-spin component one site to the right and leaves the down-spin component fixed, while $S_{\downarrow }$ shifts the down-spin component one site to the left and leaves the up-spin component fixed. \Cref{fig:splitstep} illustrates the coin and shift operations and the associated scalar embedding.

\begin{figure}[!ht]
    \centering
    \includegraphics[width=0.95\textwidth]{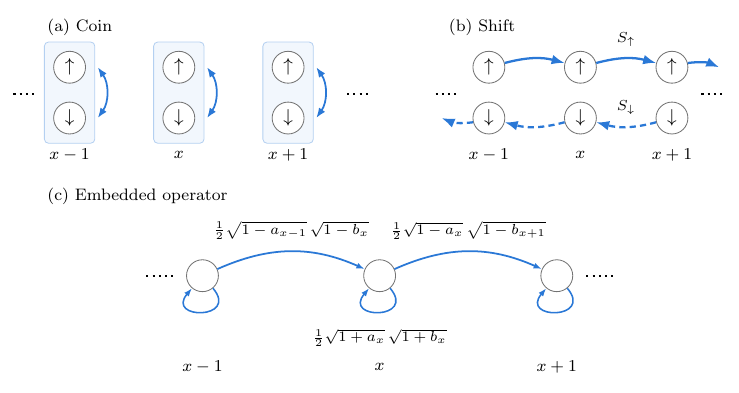}
    \caption{The split-step walk. (a) The coin $R(\vartheta(x))$ mixes the internal states
    $\ket{\uparrow}$ and $\ket{\downarrow}$ at each site $x$
    on the integer line through a site-dependent rotation.
    (b) The conditional shifts $S_{\uparrow}$ and $S_{\downarrow}$
    move the respective spin components one site right and left,
    each leaving the other component in place. One time step alternates shifts and coins, $U_{ss} =R( \vartheta _{2}) S_{\uparrow } R( \vartheta _{1}) S_{\downarrow }$. (c) The embedded operator
    $A=\Pi_{\mathfrak K}S_{\downarrow}\Pi_{\mathfrak H}$
    describes hopping with one degree of freedom per site.
    Although $A$ is not self-adjoint, $AA^*$ is a Jacobi
    (tridiagonal) matrix.}
    \label{fig:splitstep}
\end{figure}
\begin{equation*}
    S_{\uparrow } =\sum _{x} \ket{x+1}\bra{x}\otimes \ket{\uparrow }\bra{\uparrow } +\sum _{x} \ket{x}\bra{x}\otimes \ket{\downarrow }\bra{\downarrow } ,
\end{equation*}
\begin{equation*}
    S_{\downarrow } =\sum _{x} \ket{x-1}\bra{x}\otimes \ket{\downarrow }\bra{\downarrow } +\sum _{x} \ket{x}\bra{x}\otimes \ket{\uparrow }\bra{\uparrow } .
\end{equation*}
A direct computation rewrites the time evolution as
\begin{equation*}
    U_{ss} =R_{2} S_{\downarrow }^{*} R_{1} S_{\downarrow } ,
\end{equation*}
where
\begin{equation*}
    R_{1} =\sum _{x} \ket{x}\bra{x}\otimes \begin{bmatrix}
        a_{x} & \sqrt{1-a_{x}^{2}} \\
        \sqrt{1-a_{x}^{2}} & -a_{x}
    \end{bmatrix} ,\ R_{2} =\sum _{x} \ket{x}\bra{x}\otimes \begin{bmatrix}
        b_{x} & \sqrt{1-b_{x}^{2}} \\
        \sqrt{1-b_{x}^{2}} & -b_{x}
    \end{bmatrix}
\end{equation*}
with $a_{x} =\sin\left(\frac{\vartheta _{1}( x)}{2}\right) ,\ b_{x} =-\sin\left(\frac{\vartheta _{2}( x)}{2}\right)$. Take $U=S_{\downarrow }$, $\mathfrak{H} =\ker( R_{2} -I)$, $\mathfrak{K} =\ker( R_{1} -I)$ and $\theta _{k} =\phi _{k} =\pi /2$, so that
$
    R_{\mathfrak{H}}( \theta _{k}) =iR_{2} ,\ R_{\mathfrak{K}}( \phi _{k}) =iR_{1} ,
$
and therefore
\begin{equation*}
    W_{k} =R_{\mathfrak{H}}( \theta _{k}) U^{*} R_{\mathfrak{K}}( \phi _{k}) U=( iR_{2}) S_{\downarrow }^{*}( iR_{1}) S_{\downarrow } =-U_{ss} .
\end{equation*}
The projections onto $\mathfrak{H}$ and $\mathfrak{K}$ are
\begin{align*}
     & \Pi _{\mathfrak{H}} =\frac{I+R_{2}}{2} =\sum _{x}\ket{x}\bra{x} \otimes \ket{\psi _{\mathfrak{H}}( x)}\bra{\psi _{\mathfrak{H}}( x)} \\
     & \Pi _{\mathfrak{K}} =\frac{I+R_{1}}{2} =\sum _{x}\ket{x}\bra{x} \otimes \ket{\psi _{\mathfrak{K}}( x)}\bra{\psi _{\mathfrak{K}}( x)}
\end{align*}
where
\begin{equation*}
    \ket{\psi _{\mathfrak{H}}( x)} =\frac{1}{\sqrt{2}}\begin{bmatrix}
        \sqrt{1+b_{x}} \\
        \sqrt{1-b_{x}}
    \end{bmatrix} ,\ \ket{\psi _{\mathfrak{K}}( x)} =\frac{1}{\sqrt{2}}\begin{bmatrix}
        \sqrt{1+a_{x}} \\
        \sqrt{1-a_{x}}
    \end{bmatrix} .
\end{equation*}
So $U_{ss}$ is a QSVT operator, with embedded operator $A:\mathfrak{H}\rightarrow \mathfrak{K}$ given by
\begin{align*}
    A & =\Pi _{\mathfrak{K}} S_{\downarrow } \Pi _{\mathfrak{H}} \\
      & =\frac{1}{2}\left(\sum _{x}\ket{x}\bra{x+1} \otimes \sqrt{1-b_{x+1}}\sqrt{1-a_{x}}\ket{\psi _{\mathfrak{K}}( x)}\bra{\psi _{\mathfrak{H}}( x+1)}\right) \\
      & \ \ \ \ \ \ +\frac{1}{2}\left(\sum _{x}\ket{x}\bra{x} \otimes \sqrt{1+b_{x}}\sqrt{1+a_{x}}\ket{\psi _{\mathfrak{K}}( x)}\bra{\psi _{\mathfrak{H}}( x)}\right) .
\end{align*}
Identifying $\ket{x}\ket{\psi _{\mathfrak{K}}( x)}$ with $\ket{x}$ and $\bra{x}\bra{\psi _{\mathfrak{H}}( x)}$ with $\bra{x}$ turns this into
\begin{equation}\label{eq:ss-A}
    A=\sum _{x} d_{x}\ket{x}\bra{x} +\sum _{x} c_{x}\ket{x}\bra{x+1} ,
\end{equation}
where
\begin{equation}\label{eq:ss-cd}
    c_{x} :=\tfrac{1}{2}\sqrt{( 1-a_{x})( 1-b_{x+1})} ,\qquad
    d_{x} :=\tfrac{1}{2}\sqrt{( 1+a_{x})( 1+b_{x})} .
\end{equation}
Thus $A$ is a weighted shift plus a diagonal, and
\begin{equation*}
    AA^{*} =\sum _{x}\left( c_{x}^{2} +d_{x}^{2}\right)\ket{x}\bra{x} +\sum _{x} c_{x} d_{x+1}\left(\ket{x}\bra{x+1} +\ket{x+1}\bra{x}\right)
\end{equation*}
is a Jacobi matrix: self-adjoint and tridiagonal, with one degree of freedom per site rather than two. \Cref{thm:SMT} thus reduces the spectral problem for $U_{ss}$ to that of a Jacobi matrix, a classical and well-understood object. The model has been analysed in \cite{SS1,SS2} through the spectral mapping theorem of \cite{SMT}. Eigenvalues of space-inhomogeneous walks with two-phase, defect or periodic coins have also been computed by transfer matrices \cite{Kiumi2021-yg,Kiumi2022,Kiumi2022-pd,Kiumi2022-ts}. Letting the phases depend on the step turns it into
a time-dependent walk not covered by these previous studies:
\begin{equation*}
    W_{k} =\mathcal{C}_{\theta _{k}} S_{\downarrow }^{*} \mathcal{C}_{\phi _{k}} S_{\downarrow } ,
\end{equation*}
where
$
    \mathcal{C}_{\theta _{k}} :=e^{i\theta _{k}} \Pi _{\mathfrak{H}} +e^{-i\theta _{k}} (I-\Pi _{\mathfrak{H}} )$ and $\mathcal{C}_{\phi _{k}} :=e^{i\phi _{k}} \Pi _{\mathfrak{K}} +e^{-i\phi _{k}} (I-\Pi _{\mathfrak{K}} ).
$

\paragraph{The topological index.}
Consider an asymptotically homogeneous walk with
$\vartheta_j(x)\in(-\pi,\pi)$ and
$\vartheta_j(x)\to\vartheta_j^\pm\in(-\pi,\pi)$
as $x\to\pm\infty$.
Since $c_x,d_x>0$ in \eqref{eq:ss-cd}, we have the following recursions:
\begin{equation}\label{eq:ss-recursion}
    \psi\in\ker A\iff\psi_{x+1}=-\frac{d_x}{c_x}\psi_x,
    \qquad
    f\in\ker A^*\iff f_x=-\frac{c_{x-1}}{d_x}f_{x-1}.
\end{equation}
Each solution is determined by one component, so both kernels
have dimension at most one. Square summability is determined,
away from the limiting ratio $1$, by
\begin{equation*}
    r_\pm:=\lim_{x\to\pm\infty}\frac{d_x}{c_x}
    =\frac{\tan(\pi/4+\vartheta_1^\pm/4)}
          {\tan(\pi/4+\vartheta_2^\pm/4)}.
\end{equation*}
If $r_+,r_-\neq1$, then $\ker A\neq\{0\}$ exactly when
$r_+<1<r_-$, and $\ker A^*\neq\{0\}$ exactly when
$r_-<1<r_+$. The corresponding states decay exponentially
at both ends. Since $\vartheta\mapsto\tan(\pi/4+\vartheta/4)$
is increasing on $(-\pi,\pi)$, $r_\pm<1$ is equivalent to
$\vartheta_1^\pm<\vartheta_2^\pm$.
Similarly, for $B$ the same argument applies with
$\vartheta_2^\pm$ replaced by $-\vartheta_2^\pm$. To establish Fredholmness, put $c_\pm=\lim_{x\to\pm\infty}c_x$
and $d_\pm=\lim_{x\to\pm\infty}d_x$.
Cutting the lattice into two half-lines and replacing the Jacobi
coefficients of $AA^*$ by their limits is a compact perturbation:
the cut has finite rank, and the coefficient differences tend to zero.
By the standard spectral formula for constant-coefficient
half-line Jacobi operators and invariance of the essential
spectrum under compact perturbations,
\[
    \sigma_{\operatorname{ess}}(AA^*)
    =[(d_--c_-)^2,(d_-+c_-)^2]
     \cup[(d_+-c_+)^2,(d_++c_+)^2].
\]
The kernel recursions give finite-dimensional kernels for $A,B$
and their adjoints. By \cref{lem:fredholm-endpoints}, $A$ and $B$
are therefore Fredholm exactly when these bands avoid $0$ and $1$,
respectively.
For $A$, this means $d_\pm\neq c_\pm$, or $r_\pm\neq1$.
Consequently, $A$ and $B$ are Fredholm precisely when
$\vartheta_1^\pm\neq\vartheta_2^\pm$ and
$\vartheta_1^\pm\neq-\vartheta_2^\pm$, respectively,
at both ends (\cref{fig:phasedesign}). Under these respective Fredholm conditions, subtracting the
kernel dimensions obtained above gives
\begin{equation}\label{eq:ss-index}
\begin{aligned}
    \ind A&=\tfrac12\bigl[\sgn(\vartheta_2^+-\vartheta_1^+)
                      -\sgn(\vartheta_2^--\vartheta_1^-)\bigr],\\
    \ind B&=\tfrac12\bigl[\sgn(\vartheta_1^-+\vartheta_2^-)
                      -\sgn(\vartheta_1^++\vartheta_2^+)\bigr].
\end{aligned}
\end{equation}

By \cref{thm:index}, coin perturbations tending to zero
at infinity preserve these indices, including changes
on finitely many sites. Under
$\ket{x}\ket{\psi_{\mathfrak H}(x)}\mapsto\ket{x}$
and
$\ket{x}\ket{\psi_{\mathfrak K}(x)}\mapsto\ket{x}$,
the perturbed and original $A$ differ by a compact operator
on $\ell^2(\mathbb Z)$, preserving the index.
The same holds for $B$, using the orthogonal source vectors. Along continuous deformations,
the indices can change only when Fredholmness is lost:
here this occurs at the phase boundaries
$\vartheta_1^\pm=\pm\vartheta_2^\pm$,
where the gap closes.

The topological meaning of the index becomes explicit through bulk winding numbers. For a finite-range operator on $\ell^2(\mathbb Z)$ whose coefficients converge at both ends, cutting the lattice at the origin leaves a compact perturbation of the direct sum of the two limiting half-line Toeplitz operators. Under the Fredholm condition for $A$, with the Fourier convention $\widehat\psi(k)=\sum_xe^{-ikx}\psi_x$, the two limiting translation-invariant operators act as multiplication by the symbols
\[
    q_\pm(k)=d_\pm+c_\pm e^{ik},\qquad 0\leqslant k\leqslant2\pi,
\]
each tracing a counterclockwise circle of radius $c_\pm$ centred at $d_\pm$, which avoids the origin since $r_\pm\neq1$. The winding numbers
\[
    \nu_\pm:=\frac{1}{2\pi i}\int_0^{2\pi}q_\pm(k)^{-1}\partial_k q_\pm(k)\,dk
    =\begin{cases}1,&r_\pm<1,\\0,&r_\pm>1,\end{cases}
\]
determine the index through the Gohberg--Krein theorem \cite{BaumVanErpToeplitz} and invariance under compact perturbations:
\begin{equation}\label{eq:ss-winding}
    \ind A=\nu_+-\nu_-.
\end{equation}
This agrees with \eqref{eq:ss-index}, obtained independently from the kernel recursions.
 The same winding-number calculation applies to $B$ with $\vartheta_2$ replaced by $-\vartheta_2$. Both indices are thus expressed as the difference between the two bulk windings, the mechanism familiar from index formulas for split-step walks \cite{MatsuzawaIndex}. Here the windings belong to the embedded operator. For $A$, when they differ, the kernel recursions give an exponentially localised state in $\ker A$ or $\ker A^*$. This connects our embedded-operator description with the winding-number and Berry-phase descriptions of one-dimensional quantum-walk topology \cite{Kitagawa}. Changing the QSVT phases preserves this state and the winding mismatch while controlling its eigenvalue and the complementary spectrum.

For the original split-step walk, that is, $n=1$ and $\theta_1=\phi_1=\pi/2$ with $W_1=-U_{ss}$, the kernel sectors sit at the $U_{ss}$-eigenvalues $-1$ and $+1$, so $U_{ss}$ has a $\pi$-mode of multiplicity $|\ind A|$ and a $0$-mode of multiplicity $|\ind B|$, by \cref{thm:SMT} and since $\ker A$ and $\ker A^{*}$ cannot both be non-trivial. These are the known topological numbers and phase boundaries $\vartheta_1=\pm\vartheta_2$ of the model \cite{Kitagawa,SS1,SS2}, here obtained from the first-order recursion \eqref{eq:ss-recursion}.

\begin{figure}[!ht]
    \centering
    \includegraphics[width=\textwidth]{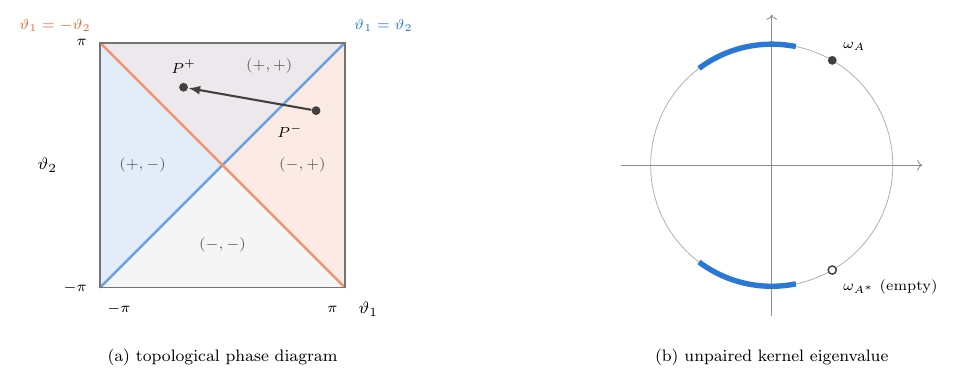}
    \caption{Topology and spectral asymmetry of the split-step walk
with $\ind A=1$ and $\ind B=0$.
(a) The asymptotic angles
$P^\pm=(\vartheta_1^\pm,\vartheta_2^\pm)$
determine the indices via \eqref{eq:ss-index}.
The phase boundaries are $\vartheta_1=\vartheta_2$ for $A$
(blue) and $\vartheta_1=-\vartheta_2$ for $B$ (orange);
the corresponding shading indicates $\vartheta_2>\vartheta_1$
and $\vartheta_1+\vartheta_2>0$, respectively.
(b) The index measures the dimension mismatch between
conjugate kernel sectors: one state at $\omega_A$ (filled circle)
and none at $\omega_{A^*}$ (empty circle).
The QSP phase sequence $\Phi$ is chosen so that
$\omega_A\notin\{1,-1\}$.}
    \label{fig:phasedesign}
\end{figure}

\paragraph{Spectral design on the split-step walk.}
Part (a) of \cref{thm:design} is visible in every figure below: the band part is the image of $\sigma(AA^{*})$ under the designed polynomial, and the kernel eigenvalue is placed by the phases alone. For the latter, fix the angle profiles $\vartheta_1(x),\vartheta_2(x)$ with $\ind A=1$, and let $\psi\in\ker A$ be the normalised localised state. Varying the QSVT phases leaves $A$ and $\psi$ unchanged, while
$W^{(n)}\psi=e^{i\lambda}\psi$, where
$\lambda=\sum_{k=1}^n(\theta_k-\phi_k)$.
Thus the same state can be placed at any prescribed quasi-energy.

\emph{Stretching a band and separating two bands (\cref{thm:design}(b)(i),(ii)).} 
Taking $(\vartheta_1^+,\vartheta_2^+)=(-1.0,2.0)$ and $(\vartheta_1^-,\vartheta_2^-)=(2.4,1.4)$ gives $\ind A=1$, $\ind B=0$ and two disjoint bulk bands, $I_-=[0.061,0.338]$ and $I_+=[0.465,0.939]$. If their angular separation is $\delta$, \cref{thm:design}(b)(ii) bounds the output gap by $n\delta$. \Cref{fig:bandgap} illustrates three designs: a single step saturates the gap bound; repeating a suitably chosen step widens the gap at $n=2$; and a Chebyshev phase choice doubles the angular length of $I_-$, attaining the stretching bound, while folding the two band images onto each other and closing the gap. Thus the phases control whether the bands separate or overlap, as well as their lengths. Phase choices and numerical widths are given in the caption.

\begin{figure}[!ht]
    \centering
    \includegraphics[width=\textwidth]{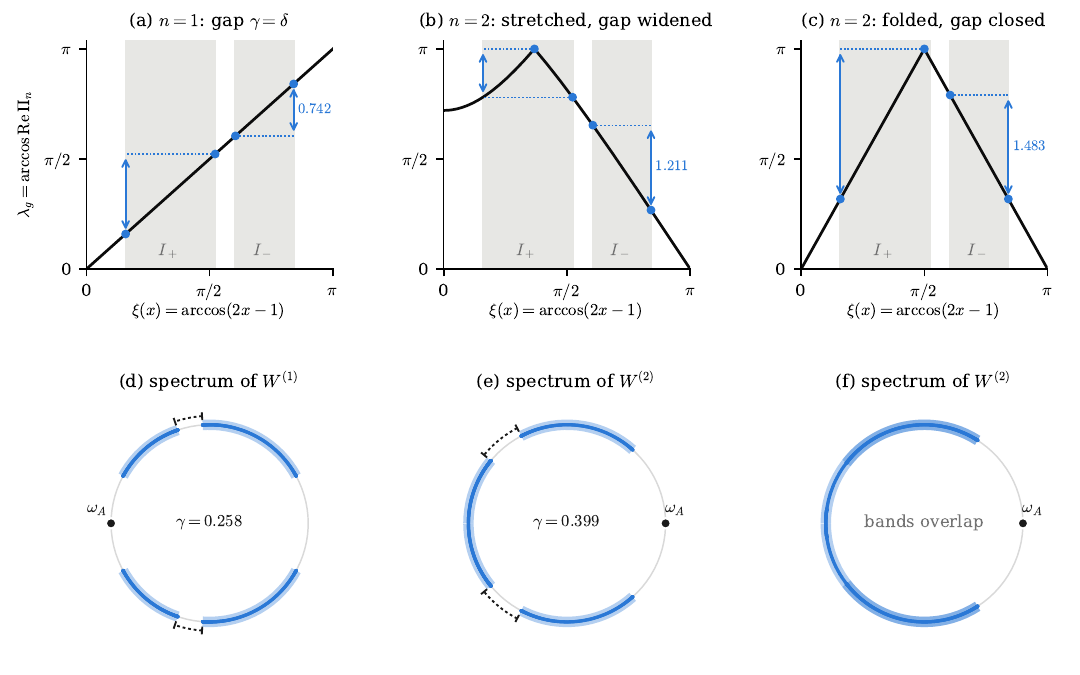}
    \caption{Control of the gap between two bands on a $360$-site ring with the bulk angles given in the text. Top: the quasi-energy $\lambda_g=\arccos\operatorname{Re}\Pi_n$ versus $\xi(x)=\arccos(2x-1)$, with $I_\pm$ shaded. The slope has absolute value at most $n$ by \cref{thm:design}(b)(i); arrows mark the output quasi-energy widths (labelled for $I_-$). Bottom: the spectrum of $W^{(n)}$, showing numerical eigenvalues (dots), analytic band images (shaded arcs), and the predicted kernel eigenvalue $\omega_A$ (filled circle).
The input angular separation is $\delta=\xi(\max I_-)-\xi(\min I_+)\simeq0.258$, and $I_-$ has angular length $\Delta_-\simeq0.742$. (a,d) $n=1$, $\theta=\pi/2$, $\phi=3\pi/2$: the gap equals $\delta$, and the image of $I_-$ has width $\Delta_-$. (b,e) $n=2$, $\theta_k=\phi_k=4.1465$: the largest gap found numerically is $0.399<2\delta\simeq0.517$, and the image of $I_-$ has width $1.211<2\Delta_-$. The image of $I_+$ folds near $z=-1$; where the two branches overlap, the infinite-line model has absolutely continuous spectrum of multiplicity two. (c,f) $n=2$, $\theta_k=3\pi/2$, $\phi_k=\pi/2$: Chebyshev folding overlaps the band images, closing the gap, while the image of $I_-$ has width $2\Delta_-\simeq1.483$, attaining the stretching bound.}
    \label{fig:bandgap}
\end{figure}

\begin{figure}[!ht]
    \centering
    \includegraphics[width=\textwidth]{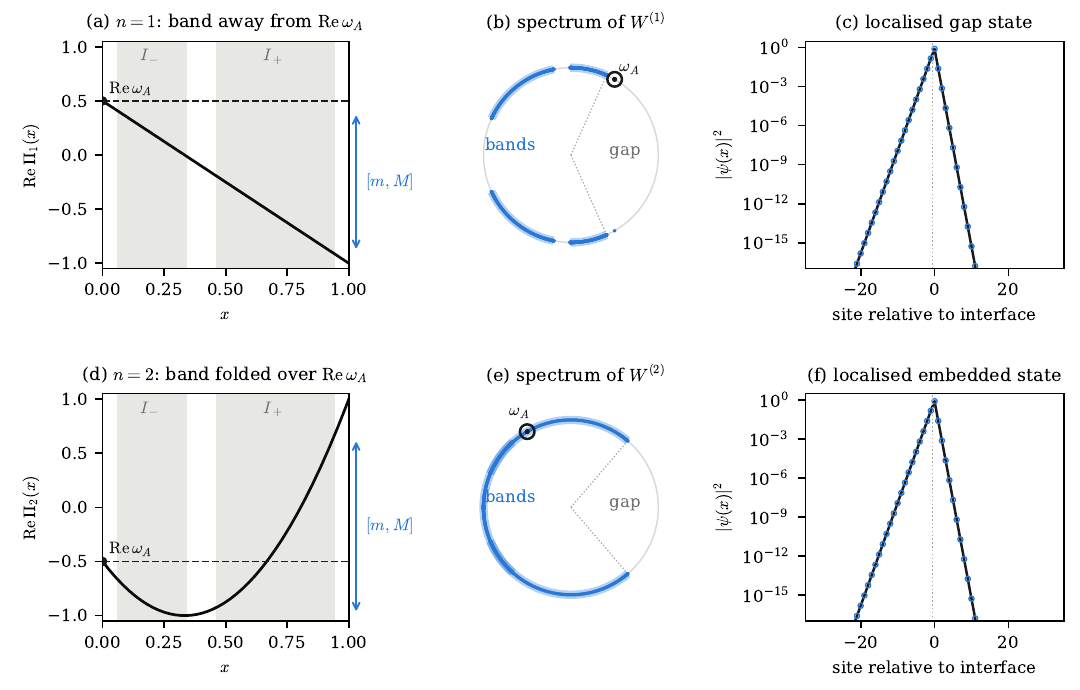}
    \caption{
A kernel state in a gap ($n=1$, top) and in a band ($n=2$, bottom), on the same $360$-site ring as \cref{fig:bandgap}. Top: $\theta=(\pi+\lambda)/2$, $\phi=(\pi-\lambda)/2$, $\lambda=\pi/3$. (a) The affine polynomial of \eqref{eq:design-n1} maps the shaded bulk bands $I_\pm$ into $[m,M]$, which excludes $\cos\lambda$. (b) Numerical eigenvalues (dots) are overlaid on the analytic bulk bands (blue arcs); the open circle marks the predicted eigenvalue $\omega_A=e^{i\lambda}$. The conjugate mode belongs to the opposite interface introduced by periodic closure. (c) Near the selected interface, the numerical site probabilities (blue circles) agree with the normalised infinite-line kernel solution from \eqref{eq:ss-recursion} (solid line); the dotted line marks the interface. Bottom: $\theta_k=2\pi/3$, $\phi_k=\pi/3$. (d) The quadratic polynomial maps the bulk bands to $[-1,0.650]$, now containing $\operatorname{Re}\omega_A=-1/2$ (dashed). (e,f) The predicted eigenvalue now lies within the band images, while the numerical state agrees with the same localised infinite-line solution. All ring eigenvalues are retained in (b,e); (c,f) show only the neighbourhood of the selected interface. The ring illustrates the infinite-line BIC; its own spectrum is purely discrete. Numerical checks: for $\lambda\in\{0.4,0.8,\pi/3,1.6,2.2\}$, the gap-mode eigenphases agree with $\pm\lambda$ to within $10^{-14}$. For the embedded modes, rings of $120$ and $360$ sites give the predicted eigenvalue to within $10^{-14}$; random coin-angle changes of up to $\pm1.2$ on five sites at each interface shift it by less than $10^{-14}$.}
    \label{fig:kerneldemo}
\end{figure}
\emph{Separating the kernel eigenvalue from the band (\cref{thm:design}(b)(iii)).} The design is explicit already at $n=1$ (\cref{ex:homo}). By \cref{thm:poly}, $\Pi _{1}( x) =e^{i\theta }\left( e^{-i\phi } +2i\sin \phi \,x\right)$, so that
\begin{equation}\label{eq:design-n1}
    \omega _{A} =\Pi _{1}( 0) =e^{i( \theta -\phi )} ,\qquad
    \operatorname{Re} \Pi _{1}( x) =\cos( \theta -\phi ) -2\sin \theta \sin \phi \,x .
\end{equation}
Choose $\theta-\phi=\lambda$ with $\sin\theta\sin\phi\neq0$. Since $g(x)=\operatorname{Re}\Pi_1(x)$ is non-constant affine and $g(0)=\cos\lambda$, the bulk spectrum stays away from the kernel eigenvalue. For $0<\lambda<\pi$, an explicit choice is $\theta=(\pi+\lambda)/2$ and $\phi=(\pi-\lambda)/2$, which gives $g(x)=\cos\lambda-(1+\cos\lambda)x$. This is \cref{thm:design}(b)(iii) with $n=1$: the quasi-energy gap between the kernel eigenvalue and the band is at most $d=\arccos(1-2x_-)$, where $x_-=\min I_-$. At degree $n$, the gap $nd$ is attained by $g(x)=T_n(1-2x)$ under the conditions of \cref{thm:design}(b)(iii).

\Cref{fig:kerneldemo}(a--c) compares the analytic band and kernel-state predictions with numerical diagonalisation on a finite ring. Near the selected interface, the numerical eigenstate at $\omega_A$ agrees with the infinite-line kernel solution. Periodic closure introduces a second, oppositely oriented interface and a conjugate mode, which is retained in the spectrum plot.

\phantomsection\label{par:bic}\emph{Embedding the kernel eigenvalue (\cref{thm:design}(b)(iv)): a bound state in the continuum.} A bound state in the continuum (BIC) is a localised eigenstate whose energy lies within a continuum of extended states: localisation and propagation coexist at the same energy \cite{BICreview}. Established mechanisms include symmetry and destructive interference, as well as constructions involving band topology \cite{YangBahramyNagaosa,BenalcazarCerjan}. Embedded eigenvalues have also been studied in quantum walks with defects \cite{MoriokaSegawa,KiumiThreeState}. Keeping the embedded operator and its localised kernel state fixed, we choose a degree-two polynomial that places the state's quasi-energy inside an absolutely continuous band of the split-step evolution.

To establish the BIC on the infinite lattice $\mathbb Z$, choose the coin angles to equal $(\vartheta_1^-,\vartheta_2^-)$ for $x<0$ and $(\vartheta_1^+,\vartheta_2^+)$ for $x\geqslant0$. Away from the interface, $AA^*$ has constant Jacobi coefficients $c_\pm^2+d_\pm^2$ on the diagonal and $c_\pm d_\pm>0$ off the diagonal. Cutting the interface and adjusting the boundary diagonal entry gives two constant-coefficient half-line Jacobi operators with absolutely continuous bands
\[
    I_\pm=[(d_\pm-c_\pm)^2,(d_\pm+c_\pm)^2].
\]
Restoring the interface is a finite-rank perturbation, so the absolutely continuous parts are preserved by the Kato--Rosenblum theorem \cite{KatoTraceClass}. Thus $AA^{*}|_{(0,1)}$ has absolutely continuous spectrum $I_-\cup I_+$. 
These are the bands $I_\pm$ specified above for \cref{fig:bandgap}. Take the fixed phases $\theta_1=\theta_2=2\pi/3$ and $\phi_1=\phi_2=\pi/3$, so that $W^{(2)}=W_1^2$ and, by \cref{ex:homo},
\[
    g(x)=2\left(\tfrac12-\tfrac32x\right)^2-1,
    \qquad g(0)=g(2/3)=-\tfrac12,\qquad g'(2/3)=3 .
\]
This is the polynomial of \cref{thm:design}(b)(iv) with $x_1=2/3$ and $r=2/3$; since $x_1$ is an interior point of the absolutely continuous band $I_+$, $\omega_A$ and $\omega_{A^*}$ lie in the interior of absolutely continuous arcs of $W^{(2)}|_{\mathcal L_{(0,1)}}$. The first-order recursion \eqref{eq:ss-recursion} gives an exponentially localised kernel state since $r_+<1<r_-$, proving a BIC for this infinite-line model. The finite rings below illustrate this result; a finite ring itself has only point spectrum.

\Cref{fig:kerneldemo}(d--f) shows that the numerical state continues to agree with the same infinite-line kernel solution after spectral folding places its eigenvalue inside the band images. On the infinite line, finite-support coin perturbations change the identified scalar operators $A$ and $AA^*$ by finite-rank operators, preserving the index and the absolutely continuous bands. At fixed QSVT phases, the BIC therefore persists.

\section{Conclusion and discussion\label{sec:conclusion}}

Operator QSVT reveals a clean division of roles between the QSP phases and the block encoding, mirrored by the decomposition of the Hilbert space into the qubitized sector and the four kernel sectors. On the qubitized sector, the $SU(2)$ structure of QSP enforces conjugation symmetry and reduces spectral design across all spectral types to choosing a single real polynomial, subject to explicit degree bounds on band stretching and gap opening. On the kernel sectors---where standard qubitization degenerates---the QSP phases act only as scalar phase shifts that place the eigenvalues, while the Fredholm indices of the encoding blocks fix the net number of kernel eigenstates against compact perturbations and continuous deformations.

The two quantum-walk models realise this division of roles from opposite sides. In Szegedy's walk, the algorithmic dynamics live in the qubitized sector through the Markov-chain discriminant, while the discarded kernel sectors on a finite graph carry its Euler characteristic. In the split-step walk, $AA^*$ is a scalar Jacobi operator, while the kernel recursions for $A$ and $B$ identify their indices with differences of bulk winding numbers; QSP phase design then places the localised interface states at arbitrary quasi-energies and, already at two steps, folds an absolutely continuous band over the protected eigenvalue to produce a bound state in the continuum.

Disordered walks with non-convergent coin profiles provide another setting for the framework, provided the relevant encoding blocks remain Fredholm. Connecting their QSVT spectral design with existing quantum-walk index theory \cite{CGSVWW,CGGSVWW,CGSVW18,CGWW21,MTW} is a direction for further study.

Extending this framework to two-dimensional quantum walks \cite{Kitagawa} is a natural next step. There, bulk topology is governed by Chern numbers rather than winding numbers, and their established connection with Fredholm indices \cite{BES94} suggests seeking block encodings whose kernel sectors capture chiral edge currents while QSP phases reshape the bulk bands. A related question for higher-dimensional BICs is how the degree-two spectral-folding mechanism interacts with protection by Chern numbers or crystalline symmetries \cite{YangBahramyNagaosa,BenalcazarCerjan}, and which spatially extended perturbations preserve both the Fredholm gap of $A$ and the surrounding continuous spectrum.

Finally, our results suggest a new role for topology in
quantum algorithms. Existing algorithms estimate Berry phases
\cite{BerryPhaseComplexity,AdiabaticCancellation,QuantizedBerryComplexity},
Chern numbers via QSP \cite{ChernEstimation}, and Betti
numbers \cite{LGZ,Hayakawa2022quantumalgorithm}.
Here topology enters as a structural property of the block
encoding itself: a non-zero Fredholm index protects the
existence of kernel eigenstates even when QSP places their
eigenvalues inside continuous bands.
Exploring whether the topology of a block encoding can
protect computationally useful states without spectral
isolation is an interesting direction for future work.

\section*{Acknowledgements}
C. Kiumi was supported by JST PRESTO, Grant Number JPMJPR25F1; JSPS KAKENHI, Grant Numbers JP22KJ1408 and JP26K17052; and JST ASPIRE, Grant Number JPMJAP2319.
\bibliographystyle{quantum}

\appendix
\crefalias{section}{appendix}
\section{Proof of \cref{thm:poly}: Polynomial transformation}\label{app:poly}
\paragraph{(a) Even iteration.}
Writing $R_{\mathfrak K}(\phi_k)=\operatorname{diag}(e^{i\phi_k}I_{\mathfrak K},e^{-i\phi_k}I_{\mathfrak K^\perp})$ on $\mathfrak K\oplus\mathfrak K^\perp$ and using $C^*C=I_{\mathfrak H}-A^*A$, $C^*D=-A^*B$, $D^*C=-B^*A$ and $B^*B=I_{\mathfrak H^\perp}-D^*D$, we obtain
\[
    U^*R_{\mathfrak K}(\phi_k)U=\begin{bmatrix}
        e^{-i\phi_k}I_{\mathfrak H}+2i\sin\phi_k\,A^*A & 2i\sin\phi_k\,A^*B \\
        2i\sin\phi_k\,B^*A & e^{i\phi_k}I_{\mathfrak H^\perp}-2i\sin\phi_k\,D^*D
    \end{bmatrix}.
\]
Multiplying on the left by $R_{\mathfrak H}(\theta_k)=\operatorname{diag}(e^{i\theta_k}I_{\mathfrak H},e^{-i\theta_k}I_{\mathfrak H^\perp})$ gives \eqref{eq:Wk}, which establishes the base case $n=1$ with $\Pi_0\equiv1$ and $\Phi_0\equiv0$. We assume that
\begin{equation*}
    W^{( n-1)} =\begin{bmatrix}
        \Pi _{n-1}\left( A^{*} A\right) & i\Phi _{n-1}\left( A^{*} A\right) A^{*} B \\
        i\Phi _{n-1}^{*}\left( D^{*} D\right) B^{*} A & \Pi _{n-1}^{*}\left( D^{*} D\right)
    \end{bmatrix}
\end{equation*}
where $\Pi _{n-1} ,\Phi _{n-1} \in \mathbb{C}[ X]$. The representation is preserved under multiplication:
\begin{align*}
    W^{( n)} & =W_{n} W^{( n-1)} \\
                           & =\begin{bmatrix}
                                  P_{n}\left( A^{*} A\right) & iQ_{n}\left( A^{*} A\right) A^{*} B \\
                                  iQ_{n}^{*}\left( D^{*} D\right) B^{*} A & P_{n}^{*}\left( D^{*} D\right)
                              \end{bmatrix}\begin{bmatrix}
                                               \Pi _{n-1}\left( A^{*} A\right) & i\Phi _{n-1}\left( A^{*} A\right) A^{*} B \\
                                               i\Phi _{n-1}^{*}\left( D^{*} D\right) B^{*} A & \Pi _{n-1}^{*}\left( D^{*} D\right)
                                           \end{bmatrix} \\
                           & =\begin{bmatrix}
                                  \Pi _{n}\left( A^{*} A\right) & i\Phi _{n}\left( A^{*} A\right) A^{*} B \\
                                  i\Phi _{n}^{*}\left( D^{*} D\right) B^{*} A & \Pi _{n}^{*}\left( D^{*} D\right)
                              \end{bmatrix}
\end{align*}
where
\begin{align*}
     & \Pi _{n}( x) =P_{n}( x) \Pi _{n-1}( x) -Q_{n}( x) \Phi _{n-1}^{*}( x)( 1-x) x, \\
     & \Phi _{n}( x) =P_{n}( x) \Phi _{n-1}( x) +Q_{n}( x) \Pi _{n-1}^{*}( x) .
\end{align*}
Here we repeatedly use the unitarity relations. For example, for the top-left entry $\Pi _{n}\left( A^{*} A\right)$ we use
\begin{equation*}
    \begin{aligned}
        A^{*} B\Phi _{n-1}^{*}\left( D^{*} D\right) B^{*} A &= A^{*} B\Phi _{n-1}^{*}\left( I-B^{*} B\right) B^{*} A \\
        &= A^{*} \Phi _{n-1}^{*}\left( I-BB^{*}\right) BB^{*} A \\
        &= A^{*} \Phi _{n-1}^{*}\left( AA^{*}\right) BB^{*} A \\
        &= \Phi _{n-1}^{*}\left( A^{*} A\right)\left( I-A^{*} A\right) A^{*} A,
    \end{aligned}
\end{equation*}
and for the bottom-right entry $\Pi _{n}^{*}\left( D^{*} D\right)$ we use
\begin{align*}
    B^{*} A\Phi _{n-1}\left( A^{*} A\right) A^{*} B & =B^{*} \Phi _{n-1}\left( AA^{*}\right) AA^{*} B \\
                                                    & =\Phi _{n-1}\left( I-B^{*} B\right) B^{*} AA^{*} B \\
                                                    & =\Phi _{n-1}\left( D^{*} D\right) B^{*} B\left( I-B^{*} B\right) \\
                                                    & =\Phi _{n-1}\left( D^{*} D\right)\left( I-D^{*} D\right) D^{*} D.
\end{align*}
The remaining entries are obtained in the same way. Therefore, we obtain
\begin{equation*}
    \begin{bmatrix}
        \Pi _{n}( x) \\
        \Phi _{n}^{*}( x)
    \end{bmatrix} =\begin{bmatrix}
        P_{n}( x) & -Q_{n}( x)( 1-x) x \\
        Q_{n}^{*}( x) & P_{n}^{*}( x)
    \end{bmatrix}\begin{bmatrix}
        \Pi _{n-1}( x) \\
        \Phi _{n-1}^{*}( x)
    \end{bmatrix} .
\end{equation*}

\paragraph{(b) Odd iteration.}
We directly calculate the following multiplication,
\begin{align*}
    R_{\mathfrak{K}}( \phi _{n+1}) UW^{( n)} & =\begin{bmatrix}
                                                                  e^{i\phi _{n+1}} A & e^{i\phi _{n+1}} B \\
                                                                  e^{-i\phi _{n+1}} C & e^{-i\phi _{n+1}} D
                                                              \end{bmatrix}\begin{bmatrix}
                                                                               \Pi _{n}\left( A^{*} A\right) & i\Phi _{n}\left( A^{*} A\right) A^{*} B \\
                                                                               i\Phi _{n}^{*}\left( D^{*} D\right) B^{*} A & \Pi _{n}^{*}\left( D^{*} D\right)
                                                                           \end{bmatrix}.
\end{align*}
Up to the phase factors $e^{\pm i\phi _{n+1}}$, the top-left entry is
\begin{align*}
     & A\Pi _{n}\left( A^{*} A\right) +Bi\Phi _{n}^{*}\left( D^{*} D\right) B^{*} A \\
     & =\left( \Pi _{n}\left( AA^{*}\right) +Bi\Phi _{n}^{*}\left( I-B^{*} B\right) B^{*}\right) A \\
     & =\left( \Pi _{n}\left( AA^{*}\right) +i\Phi _{n}^{*}\left( I-BB^{*}\right) BB^{*}\right) A \\
     & =\left( \Pi _{n}\left( AA^{*}\right) +i\Phi _{n}^{*}\left( AA^{*}\right)\left( I-AA^{*}\right)\right) A,
\end{align*}
and the bottom-right entry is
\begin{align*}
     & Ci\Phi _{n}\left( A^{*} A\right) A^{*} B+D\Pi _{n}^{*}\left( D^{*} D\right) \\
     & =-Ci\Phi _{n}\left( I-C^{*} C\right) C^{*} D+\Pi _{n}^{*}\left( DD^{*}\right) D \\
     & =\left( -i\Phi _{n}\left( I-CC^{*}\right) CC^{*} +\Pi _{n}^{*}\left( DD^{*}\right)\right) D \\
     & =\left( -i\Phi _{n}\left( DD^{*}\right)\left( I-DD^{*}\right) +\Pi _{n}^{*}\left( DD^{*}\right)\right) D.
\end{align*}
The remaining entries are obtained in the same way.

\paragraph{(c) QSP correspondence.\label{app:qsp}}
Parts (a) and (b) hold for any choice of $(\mathcal{H} ,\Pi _{\mathfrak{H}} ,\Pi _{\mathfrak{K}} ,U)$, and the polynomials $\Pi _{n} ,\Phi _{n} ,\Theta _{n} ,\Omega _{n}$ appearing in them depend only on the phases. To determine these polynomials explicitly, we substitute the specific two-dimensional example
\[
    \mathcal{H} =\mathbb{C}^{2} ,\qquad
    \Pi _{\mathfrak{H}} =\Pi _{\mathfrak{K}} =\begin{bmatrix}
        1 & 0 \\
        0 & 0
    \end{bmatrix} ,\qquad
    U=R( \sigma ) =\begin{bmatrix}
        \sigma & \sqrt{1-\sigma ^{2}} \\
        \sqrt{1-\sigma ^{2}} & -\sigma
    \end{bmatrix} \quad ( \sigma \in ( -1,1)) .
\]
For this choice, $R_{\mathfrak{H}}( \varphi ) =R_{\mathfrak{K}}( \varphi ) =e^{i\varphi Z}$ and $U^{*} =R( \sigma )$, so the QSVT iterates $W^{( n)}$ and $R_{\mathfrak{K}}( \phi _{n+1}) UW^{( n)}$ reduce directly by definition to the QSP matrix products \eqref{eq:QSPseq}:
\[
    W^{( n)} =\prod _{k=n}^{1}\bigl( e^{i\theta _{k} Z} R( \sigma ) e^{i\phi _{k} Z} R( \sigma )\bigr) =S_{\Phi }( \sigma ) ,\qquad
    R_{\mathfrak{K}}( \phi _{n+1}) UW^{( n)} =e^{i\phi _{n+1} Z} R( \sigma ) S_{\Phi }( \sigma ) =S_{\Phi '}( \sigma ) .
\]
On the other hand, the blocks of $U=R( \sigma )$ are the $1\times 1$ scalars $A=\sigma$, $B=C=\sqrt{1-\sigma ^{2}}$ and $D=-\sigma$, with $A^{*} A=AA^{*} =D^{*} D=DD^{*} =\sigma ^{2}$ and $A^{*} B=B^{*} A=\sigma \sqrt{1-\sigma ^{2}}$. Substituting these scalars into the block formulas of parts (a) and (b) gives
\begin{align*}
    S_{\Phi }( \sigma ) &=\begin{bmatrix}
        \Pi _{n}\left( \sigma ^{2}\right) & i\Phi _{n}\left( \sigma ^{2}\right) \sigma \sqrt{1-\sigma ^{2}} \\
        i\Phi _{n}^{*}\left( \sigma ^{2}\right) \sigma \sqrt{1-\sigma ^{2}} & \Pi _{n}^{*}\left( \sigma ^{2}\right)
    \end{bmatrix} ,\\
    S_{\Phi '}( \sigma ) &=\begin{bmatrix}
        \Theta _{n}\left( \sigma ^{2}\right) \sigma & \Omega _{n}\left( \sigma ^{2}\right)\sqrt{1-\sigma ^{2}} \\
        \Omega _{n}^{*}\left( \sigma ^{2}\right)\sqrt{1-\sigma ^{2}} & -\Theta _{n}^{*}\left( \sigma ^{2}\right) \sigma
    \end{bmatrix} .
\end{align*}
Comparing the top rows with \eqref{eq:QSPseq} and dividing the $( 1,2)$ entries by $\sqrt{1-\sigma ^{2}} >0$ yields \eqref{eq:qsp-pair} for all $\sigma \in ( -1,1)$, and hence for all $\sigma \in \mathbb{R}$ since both sides are polynomials.

\section{Proof of \cref{prop:L-decomp}: Operator qubitization\label{app:qubitization}}

By the spectral decomposition of \(AA^*\) we have $\mathfrak K=\mathfrak K_0\oplus\mathfrak K_{(0,1)}\oplus\mathfrak K_1$. Since \(A^*\mathfrak K_0=\{0\}\) and \(B^*\mathfrak K_1=\{0\}\), and since \(A^*\) is isometric on \(\mathfrak K_1\) and \(B^*\) is isometric on \(\mathfrak K_0\), we get
\[
    \overline{\Ran A^*}
    =
    V_A\mathfrak K_{(0,1)}\oplus A^*\mathfrak K_1,
    \qquad
    \overline{\Ran B^*}
    =
    B^*\mathfrak K_0\oplus V_B\mathfrak K_{(0,1)} .
\]
The two summands in each identity are mutually orthogonal because $\langle A^{*} f,A^{*} g\rangle =\langle AA^{*} f,g\rangle =0$ whenever $f\in \mathfrak{K}_{( 0,1)}$ and $g\in \mathfrak{K}_{1}$, and similarly for $B^{*}$; regrouping gives $\mathcal{L}=\mathcal{L}_{0}\oplus\mathcal{L}_{(0,1)}\oplus\mathcal{L}_{1}$.

Next, for $f\in\ker A^{*}$ we have $BB^{*}f=(I-AA^{*})f=f$, so $D^{*}DB^{*}f=(I-B^{*}B)B^{*}f=0$ and $\mathcal{L}_{0}=B^{*}\ker A^{*}\subseteq\ker D$; conversely, if $\eta\in\ker D$, then $B^{*}B\eta=\eta$ and $A^{*}B\eta=-C^{*}D\eta=0$, so $B\eta\in\ker A^{*}$ and $\eta=B^{*}(B\eta)\in\mathcal{L}_{0}$. Thus $\mathcal{L}_{0}=\ker D$, and similarly $\mathcal{L}_{1}=\ker C$. These identifications of the kernels also appear in \cite{CGWW21}. We also check that $\mathcal{L}^{\perp } =\ker A\oplus \ker B$. For all $f,g\in \mathfrak{K}$, $\psi \in \ker A$ and $\eta \in \ker B$,
\begin{equation*}
    \langle A^{*} f,\psi\rangle +\langle B^{*} g,\eta\rangle =\langle f,A\psi \rangle +\langle g,B\eta \rangle =0.
\end{equation*}
Conversely, let $\psi \in \mathfrak{H} ,\ \eta \in \mathfrak{H}^{\perp }$ and assume
\begin{equation*}
    \langle A^{*} f,\psi \rangle +\langle B^{*} g,\eta \rangle =\langle f,A\psi \rangle +\langle g,B\eta \rangle =0
\end{equation*}
for all $f,g\in \mathfrak{K}$. Taking $g=0$ and letting $f$ range over $\mathfrak{K}$ forces $A\psi =0$; taking $f=0$ forces $B\eta =0$. This gives the displayed decompositions of $\mathcal{H}$.

For $\psi\in\ker A$ and $\eta\in\ker B$, the off-diagonal entries of $W^{(n)}$ in \cref{thm:poly}(a) vanish since $B^{*}A\psi=0$ and $A^{*}B\eta=0$, while $A^{*}A\psi=0$ and $D^{*}D\eta=(I-B^{*}B)\eta=\eta$; hence
\[
    W^{(n)}\begin{bmatrix}
        \psi \\
        \eta
    \end{bmatrix} =\begin{bmatrix}
        \Pi _{n}( 0) \psi \\
        \Pi _{n}^{*}( 1) \eta
    \end{bmatrix}\in \ker A\oplus \ker B=\mathcal{L}^{\perp } .
\]
For $f_{0} \in \mathfrak{K}_{0} =\ker A^{*}$ we have $AA^{*} f_{0} =0$ and $A^{*} f_{0} =0$, so \eqref{eq:W-on-L} gives $W^{( n)} B^{*} f_{0} =\Pi _{n}^{*}( 0) B^{*} f_{0} \in \mathcal{L}_{0}$; for $f_{1} \in \mathfrak{K}_{1} =\ker B^{*}$ we have $AA^{*} f_{1} =f_{1}$ and $B^{*} f_{1} =0$, so \eqref{eq:W-on-L} gives $W^{( n)} A^{*} f_{1} =\Pi _{n}( 1) A^{*} f_{1} \in \mathcal{L}_{1}$. Thus $\mathcal{L}_{0}$, $\mathcal{L}_{1}$ and $\mathcal{L}^{\perp }$ are invariant under $W^{( n)}$. For $\mathcal L_{(0,1)}$, on $\mathfrak K_{(0,1)}$ one has $P(A^*A)V_A=V_AP(AA^*)$ and $P(D^*D)V_B=V_BP(AA^*)$ for every polynomial $P$, together with
\[
    A^*BV_B=V_A(AA^*)^{1/2}(I-AA^*)^{1/2},
    \qquad
    B^*AV_A=V_B(AA^*)^{1/2}(I-AA^*)^{1/2};
\]
substituting these into the block form of \cref{thm:poly} gives $W^{(n)}J=J\,M_n(AA^*|_{(0,1)})$, so $\mathcal L_{(0,1)}$ is invariant under $W^{(n)}$ and the stated unitary equivalence holds. Finally, the factorisation \eqref{eq:qsp-fibre} of $M_n(\sigma^2)$ was already established in \cref{app:poly}(c) for $\sigma\in(-1,1)$, and extends to $\sigma\in[0,1]$ by continuity.

\section{Proof of \cref{thm:SMT}: Spectral mapping theorem\label{app:smt}}

\paragraph{(a) Kernel sectors.}
By \cref{prop:L-decomp}, $\mathcal{H}=\mathcal{L}_{0}\oplus\mathcal{L}_{(0,1)}\oplus\mathcal{L}_{1}\oplus\ker A\oplus\ker B$ is an orthogonal direct sum of $W^{(n)}$-invariant subspaces, so \eqref{eq:SMT} holds once the spectrum on the four kernel sectors $\ker A$, $\mathcal{L}_{0}=B^{*}\ker A^{*}$, $\ker B$ and $\mathcal{L}_{1}=A^{*}\ker B^{*}$ is identified with $\Sigma_{\bullet}^{\ker}$. As shown in \cref{app:qubitization}, $W^{(n)}$ acts on these four sectors as scalar multiplication by $\Pi_{n}(0)$, $\Pi_{n}^{*}(0)$, $\Pi_{n}^{*}(1)$ and $\Pi_{n}(1)$, respectively. Setting $\sigma=0$ and $\sigma=1$ in $M_{n}(\sigma^{2})=S_{\Phi}(\sigma)$ of \eqref{eq:qsp-fibre} gives
\[
    \Pi_{n}(0)=\prod_{k=1}^{n}e^{i(\theta_{k}-\phi_{k})}=\omega_{A},\qquad
    \Pi_{n}^{*}(1)=\prod_{k=1}^{n}e^{-i(\theta_{k}+\phi_{k})}=\omega_{B},
\]
and hence $\Pi_{n}^{*}(0)=\overline{\omega_{A}}=\omega_{A^{*}}$ and $\Pi_{n}(1)=\overline{\omega_{B}}=\omega_{B^{*}}$, establishing the four eigenvalue relations in \cref{thm:SMT}(a). Moreover, on $\ker A^{*}=\ker AA^{*}$ we have $BB^{*}=I-AA^{*}=I$, so $\|B^{*}f\|=\|f\|$ and $\dim\mathcal{L}_{0}=\dim\ker A^{*}$; likewise $AA^{*}=I$ on $\ker B^{*}=\ker BB^{*}$ implies $\|A^{*}f\|=\|f\|$ and $\dim\mathcal{L}_{1}=\dim\ker B^{*}$. Thus the four sectors have dimensions $\dim\ker T$ ($T\in\{A,A^{*},B,B^{*}\}$). Since the scalar operator $\omega_{T}I$ on a space of dimension $\dim\ker T$ has $\sigma=\sigma_{\operatorname{p}}=\sigma_{\operatorname{pp}}=\{\omega_{T}\}$ when $\ker T\neq\{0\}$, $\sigma_{\operatorname{ess}}=\{\omega_{T}\}$ when $\dim\ker T=\infty$ (and $\emptyset$ otherwise), and $\sigma_{\operatorname{ac}}=\sigma_{\operatorname{sc}}=\emptyset$, taking the union over $T\in\{A,A^{*},B,B^{*}\}$ yields $\Sigma_{\bullet}^{\ker}$.

\paragraph{(b) Qubitized sector.}
We prove \eqref{eq:SMT-band} on $\mathcal L_{(0,1)}$, first for $\sigma$, $\sigma_{\operatorname{ess}}$, $\sigma_{\operatorname{p}}$ and then for the three spectral types. By \cref{prop:L-decomp}, $W^{(n)}|_{\mathcal{L}_{(0,1)}}$ is unitarily equivalent to $M_{n}(AA^{*}|_{(0,1)})$. For each \(x\in(0,1)\), \eqref{eq:unitarity} gives $\det M_n(x)=1$ and $\operatorname{tr}M_n(x)=2\operatorname{Re}\Pi_n(x)$, so that
\[
    z\in\sigma(M_n(x))
    \Longleftrightarrow
    z^2-2\operatorname{Re}\Pi_n(x)\,z+1=0
    \Longleftrightarrow
    z\in\mathbb T\ \text{and}\ \operatorname{Re} z
    =
    \operatorname{Re}\Pi_n(x),
\]
the last equivalence because $M_n(x)$ is unitary, so its spectrum lies on $\mathbb T$, where $z+z^{-1}=2\operatorname{Re}z$.

Since the entries of $M_{n}(AA^{*}|_{(0,1)})$ are functions of the single self-adjoint operator $AA^{*}|_{(0,1)}$, they commute, and the $2\times 2$ adjugate formula shows that $M_{n}(AA^{*}|_{(0,1)})-z$ fails to be invertible, fails to be Fredholm, or has a non-zero kernel if and only if $\operatorname{Re}\Pi_{n}(AA^{*}|_{(0,1)})-\operatorname{Re} z$ does; that is, for $z\in\mathbb{T}$,
\[
    z\in\sigma_{\bullet}\bigl(M_{n}(AA^{*}|_{(0,1)})\bigr)
    \iff
    \operatorname{Re} z\in\sigma_{\bullet}\bigl(\operatorname{Re}\Pi_{n}(AA^{*}|_{(0,1)})\bigr)
    \qquad (\sigma_{\bullet}=\sigma,\sigma_{\operatorname{ess}},\sigma_{\operatorname{p}}).
\]
For $\sigma_{\bullet}=\sigma,\sigma_{\operatorname{ess}}$, the spectral mapping theorem $\sigma_{\bullet}\bigl(\operatorname{Re}\Pi_{n}(AA^{*}|_{(0,1)})\bigr)=\operatorname{Re}\Pi_{n}\bigl(\sigma_{\bullet}(AA^{*}|_{(0,1)})\bigr)$ immediately gives \eqref{eq:SMT-band}. For $\sigma_{\bullet}=\sigma_{\operatorname{p}}$, the same identity holds when $\operatorname{Re}\Pi_{n}$ is non-constant since $\operatorname{Re}\Pi_{n}(x)=\operatorname{Re} z$ has at most finitely many roots, whereas if $\operatorname{Re}\Pi_{n}\equiv c$ and $\mathcal{L}_{(0,1)}\neq\{0\}$ then $\sigma_{\operatorname{p}}\bigl(\operatorname{Re}\Pi_{n}(AA^{*}|_{(0,1)})\bigr)=\sigma_{\operatorname{p}}(cI)=\{c\}$, giving $\Sigma_{\operatorname{p}}^{(0,1)}=\{z\in\mathbb{T}:\operatorname{Re} z=c\}$.

It remains to prove \eqref{eq:SMT-band} for $\sigma_{\operatorname{ac}}$, $\sigma_{\operatorname{sc}}$ and $\sigma_{\operatorname{pp}}$. By the spectral theorem, $AA^{*}|_{(0,1)}$ is unitarily equivalent to multiplication by $x$ on a direct sum of spaces $L^{2}((0,1),\mu)$, and diagonalising the $2\times 2$ unitary matrix $M_{n}(x)$ measurably on $(0,1)$ makes $M_{n}(AA^{*}|_{(0,1)})$ unitarily equivalent to multiplication by
\[
    \begin{bmatrix}
        \exp\bigl(i\arccos\operatorname{Re}\Pi_{n}(x)\bigr) & 0 \\
        0 & \exp\bigl(-i\arccos\operatorname{Re}\Pi_{n}(x)\bigr)
    \end{bmatrix}
    \quad\text{on }L^{2}((0,1),\mu)\oplus L^{2}((0,1),\mu).
\]
When the polynomial $\operatorname{Re}\Pi_{n}$ is non-constant, it attains the boundary values $|\operatorname{Re}\Pi_{n}(x)|=1$ or has vanishing derivative $(\operatorname{Re}\Pi_{n})'(x)=0$ at only finitely many points in $(0,1)$. Outside this finite set (which has zero mass for any continuous measure), the maps $x\mapsto\exp\bigl(\pm i\arccos\operatorname{Re}\Pi_{n}(x)\bigr)$ have non-zero derivative and are therefore $C^{1}$ diffeomorphisms onto open arcs of $\mathbb{T}$ on each complementary interval, preserving absolute continuity, singular continuity and atoms. Consequently, for each $\bullet\in\{\operatorname{ac},\operatorname{sc},\operatorname{pp}\}$, the spectrum $\sigma_{\bullet}\bigl(M_{n}(AA^{*}|_{(0,1)})\bigr)$ is the image of $\sigma_{\bullet}(AA^{*}|_{(0,1)})$ under $x\mapsto\exp\bigl(\pm i\arccos\operatorname{Re}\Pi_{n}(x)\bigr)$, namely
\[
    \sigma_{\bullet}\bigl(M_{n}(AA^{*}|_{(0,1)})\bigr)
    = \bigl\{z\in\mathbb{T}:\operatorname{Re} z\in\operatorname{Re}\Pi_{n}\bigl(\sigma_{\bullet}(AA^{*}|_{(0,1)})\bigr)\bigr\},
\]
which gives \eqref{eq:SMT-band}. Finally, if $\operatorname{Re}\Pi_{n}\equiv c$ is constant, the diagonal entries $\exp(\pm i\arccos c)$ are independent of $x$, so $M_{n}(AA^{*}|_{(0,1)})$ has pure point spectrum, giving $\Sigma_{\operatorname{pp}}^{(0,1)}=\Sigma_{\operatorname{p}}^{(0,1)}$ and $\Sigma_{\operatorname{ac}}^{(0,1)}=\Sigma_{\operatorname{sc}}^{(0,1)}=\emptyset$.

\section{Proof of \cref{thm:design}: Spectral design\label{app:design}}

\paragraph{(a) Achievability.}
By the recursion in \cref{thm:poly}(a) and \eqref{eq:unitarity}, $\Pi_{n}$ has degree at most $n$ and $|\Pi_{n}|\leqslant 1$ on $[0,1]$, so $g=\operatorname{Re}\Pi_{n}$ satisfies $\deg g\leqslant n$ and $|g|\leqslant 1$ on $[0,1]$. Conversely, for any such $g$, the polynomial $g(\sigma^{2})$ is even of degree at most $2n$ with $|g(\sigma^{2})|\leqslant 1$ on $[-1,1]$; by QSVT achievability \cite{QSVT} and \cref{thm:poly}(c), there exist phases $\Phi\in[0,2\pi)^{2n}$ with $\operatorname{Re}\Pi_{n}(\sigma^{2})=\operatorname{Re} P_{\Phi}(\sigma)=g(\sigma^{2})$, i.e.\ $\operatorname{Re}\Pi_{n}=g$. Substituting $\operatorname{Re}\Pi_{n}=g$ into \cref{thm:SMT} gives \eqref{eq:design}, where $\Sigma_{\bullet}^{(0,1)}$ depends only on $\operatorname{Re} z$ (hence is symmetric under $z\mapsto\overline{z}$) and continuity of $g$ maps each interval to a single interval.

\paragraph{(b) Degree cost.}
To establish the lower bounds in (i)--(iii), we first show that every admissible $g$ of degree $n$ satisfies
\begin{equation}\label{eq:stretch}
    \ell\bigl(g(I)\bigr) \leqslant n\,\tilde{\ell}(I) ,\qquad I \subseteq [0,1] .
\end{equation}
By the definition \eqref{eq:arclength}, $\tilde{\ell}(I)$ is the length of the interval of $\xi:=\arccos(2x-1)\in[0,\pi]$ and $\ell(g(I))$ is the length of the image of $I$ under $\arccos g(x)$. To bound $\ell(g(I))$ in terms of $\tilde{\ell}(I)$, we therefore bound the derivative $\tfrac{d}{d\xi}\arccos g$. Since $x=(1+\cos\xi)/2$, an admissible $g(x)$ of degree $n$ is a trigonometric polynomial of degree at most $n$ in $\xi$ with $|g|\leqslant 1$. The Szeg\H{o} inequality \cite{BorweinErdelyi} then gives $|\tfrac{dg}{d\xi}|\leqslant n\sqrt{1-g^{2}}$, whence
\begin{equation*}
  \Bigl|\frac{d}{d\xi}\arccos g\Bigr| = \frac{|dg/d\xi|}{\sqrt{1-g^{2}}} \leqslant n
\end{equation*}
for $|g|<1$. Integrating over $I$ (via $(1-\delta)g$ as $\delta\to 0^{+}$) yields \eqref{eq:stretch}.

(i) \emph{Reshaping a band.} If $g(I)=J=[m,M]$, then \eqref{eq:stretch} gives $\ell(J)=\ell(g(I))\leqslant n\,\tilde{\ell}(I)$, which is the lower bound in \eqref{eq:nmin}. For the upper bound, set $k:=\lceil 2\pi/\tilde{\ell}(I)\rceil$. On $I$, $k\xi$ varies by $k\tilde{\ell}(I)\geqslant 2\pi$, so $T_{k}(2x-1)=\cos k\xi$ takes all values in $[-1,1]$. Thus $g(x)=\frac{M-m}{2}T_{k}(2x-1)+\frac{M+m}{2}$ maps both $I$ and $[0,1]$ onto $[m,M]=J\subseteq[-1,1]$, making $g$ admissible of degree $k$.

(ii) \emph{Separating two bands.} If $g(I_{1})$ and $g(I_{2})$ are separated by a gap $J\subseteq[-1,1]$, then $J\subseteq g(G)$ by the intermediate value theorem, so \eqref{eq:stretch} gives $\ell(J)\leqslant\ell(g(G))\leqslant n\,\tilde{\ell}(G)$, which is the lower bound in \eqref{eq:band-gap}. For the upper bound, let $c=(x_{1}+x_{2})/2$ be the midpoint of $G=[x_{1},x_{2}]$. For any $\varepsilon\in(0,1)$, the sign-function polynomial of \cite{QSVT} (rescaled to $[0,1]$ with tolerance $\varepsilon/2$) gives an admissible $g$ of degree $O\bigl((x_{2}-x_{1})^{-1}\log(4/\varepsilon)\bigr)$ satisfying $|g(x)-\sgn(x-c)|\leqslant\varepsilon$ for $|x-c|\geqslant(x_{2}-x_{1})/2$. Then $g\leqslant -1+\varepsilon$ on $I_{1}$ and $g\geqslant 1-\varepsilon$ on $I_{2}$, which opens the gap $J=[-1+\varepsilon,1-\varepsilon]$ of arc length $\ell(J)=\pi-2\arccos(1-\varepsilon)$. Choosing
\[
    \varepsilon=1-\cos\frac{\pi-\ell(J)}{2}=2\sin^{2}\frac{\pi-\ell(J)}{4}\geqslant\frac{(\pi-\ell(J))^{2}}{2\pi^{2}}, \qquad \log\frac{4}{\varepsilon}\leqslant 2\log\frac{2\sqrt{2}\pi}{\pi-\ell(J)},
\]
gives the upper bound in \eqref{eq:band-gap}.

(iii) \emph{Separating the kernel eigenvalue from a band.} If a gap $J\subseteq[-1,1]$ separates $g(0)$ and $g(I)$, then $J\subseteq g(G)$ by the intermediate value theorem, so $\ell(J)\leqslant\ell(g(G))\leqslant n\,\tilde{\ell}(G)$ by \eqref{eq:stretch}. Conversely, for $g(x)=T_{n}(1-2x)$, setting $\eta:=\arccos(1-2x)\in[0,\pi]$ gives $g(x)=\cos n\eta$ and $g(0)=1$. On $I=[x_{-},x_{+}]$, $\eta$ ranges over $[\tilde{\ell}(G),\tilde{\ell}(G)+\tilde{\ell}(I)]$, so under the condition $n\bigl(2\tilde{\ell}(G)+\tilde{\ell}(I)\bigr)\leqslant 2\pi$,
\[
    n\eta\in\bigl[n\tilde{\ell}(G),\,2\pi-n\tilde{\ell}(G)\bigr], \qquad g(x)=\cos n\eta\leqslant\cos\bigl(n\tilde{\ell}(G)\bigr) \quad (x\in I),
\]
separating $g(I)$ from $g(0)=1$ by the gap $J=\bigl[\cos\bigl(n\tilde{\ell}(G)\bigr),1\bigr]$ of arc length $\ell(J)=n\,\tilde{\ell}(G)$. The case $x=1$ follows via $x\mapsto 1-x$.

(iv) \emph{Merging spectra.} For $n=1$, any non-constant admissible $g$ is affine and hence injective, so $g(x_{0})\neq g(x_{1})$ for $x_{0}\neq x_{1}$. For $n=2$, $g(x)=2\bigl(\tfrac{x-c}{r}\bigr)^{2}-1$ is symmetric about $c=(x_{0}+x_{1})/2$, giving $g(x_{0})=g(x_{1})$, and satisfies $|g|\leqslant 1$ on $[0,1]$ since $\max_{x\in[0,1]}|x-c|=\max(c,1-c)=r$. When $x_{0}=0$ and $x_{1}\in(x_{-},x_{+})\subseteq I$, we have
\[
    0<c=\frac{x_{1}}{2}<1-c=r, \qquad g(0)=g(x_{1})=2\Bigl(\frac{c}{r}\Bigr)^{2}-1\in(-1,1), \qquad g'(x_{1})=\frac{2x_{1}}{r^{2}}>0.
\]
Hence $g(0)=g(x_{1})$ lies in the interior of $g(I)$, and since $\operatorname{Re}\omega_{A}=\operatorname{Re}\omega_{A^{*}}=g(0)$ (\cref{thm:SMT}(a)), the kernel eigenvalues $\omega_{A},\omega_{A^{*}}$ lie in the interior of the absolutely continuous arcs of $I$ by \eqref{eq:design}.

\section{Proof of \cref{lem:fredholm-endpoints}: Fredholmness\label{app:fredholm}}
Fredholmness of $A:\mathfrak H\to\mathfrak K$ is equivalent to the two gap conditions
\begin{equation}\label{eq:fredholm-two-gaps}
    0\notin\sigma_{\operatorname{ess}}(A^*A)
    \qquad\text{and}\qquad
    0\notin\sigma_{\operatorname{ess}}(AA^*).
\end{equation}
Indeed, the first condition holds exactly when $\Ran A$ is closed and $\ker A$ is finite dimensional. Applying the same statement to $A^*$ and using the closed range theorem gives the equivalence. The same criterion applies to $B$.

The closed range theorem and the spectral theorem give
\[
    \Ran A\text{ is closed}
    \iff \sigma(AA^*)\cap(0,\varepsilon)=\emptyset
    \quad\text{for some }\varepsilon>0.
\]
The spectral decomposition into the endpoint subspaces and $\mathfrak K_{(0,1)}$ makes this equivalent to $0\notin\sigma(AA^*|_{(0,1)})$. Since $BB^*=I-AA^*$, the same argument gives
\[
    \Ran B\text{ is closed}
    \iff \sigma(AA^*)\cap(1-\varepsilon,1)=\emptyset
    \quad\text{for some }\varepsilon>0
    \iff 1\notin\sigma(AA^*|_{(0,1)}).
\]
Combining these closed-range criteria with the finite-dimensionality of the kernel and adjoint kernel gives both equivalences by the definition of Fredholmness.\hfill$\square$

\section{Proof of \cref{thm:index}: Topological invariant\label{app:index}}

\paragraph{(a) Indices and protected states.}
The blocks $A$ and $B$ depend only on $U$ and the projections, so their Fredholm indices are independent of the QSP phases and the iteration length. By \cref{thm:SMT}(a), the four kernel sectors $\ker A$, $B^*\ker A^*$, $\ker B$ and $A^*\ker B^*$ consist of eigenvectors of $W^{(n)}$ with eigenvalues $\omega_A$, $\omega_{A^*}$, $\omega_B$ and $\omega_{B^*}$ and dimensions $\dim\ker A$, $\dim\ker A^*$, $\dim\ker B$ and $\dim\ker B^*$, respectively. For either Fredholm block $T\in\{A,B\}$, write $k=\ind T=\dim\ker T-\dim\ker T^*$. Then
\[
    \dim\ker T=k+\dim\ker T^*\geqslant k\quad(k>0),
    \qquad
    \dim\ker T^*=-k+\dim\ker T\geqslant -k=|k|\quad(k<0).
\]
The corresponding sector therefore supplies at least $|k|$ eigenstates at the eigenvalue stated in (a).

\paragraph{(b) Topological protection.}
\emph{Compact perturbations.} We first identify the perturbed source and target subspaces. For orthogonal projections $P,P'$ with $\|P'-P\|<1$, set
\[
    S:=P'P+(I-P')(I-P),\qquad
    V:=S\bigl(I-(P'-P)^2\bigr)^{-1/2}.
\]
Direct multiplication gives $SP=P'S$ and
\[
    S^*S=SS^*=I-(P'-P)^2\geqslant (1-\|P'-P\|^2)I.
\]
Thus $V$ is unitary and $VP=P'V$. If $P'-P$ is compact, then $S-I=(P'-P)(2P-I)$ is compact, and continuous functional calculus shows that $\bigl(I-(P'-P)^2\bigr)^{-1/2}-I$ is compact. Hence $V-I$ is compact.

Apply this construction to the source and target projections to obtain unitaries $V_H,V_K$ on $\mathcal H$ with
\[
    V_H\Pi_{\mathfrak H}V_H^*=\Pi'_{\mathfrak H},\qquad
    V_K\Pi_{\mathfrak K}V_K^*=\Pi'_{\mathfrak K},
\]
and $V_H-I,V_K-I$ compact. The same $V_H$ identifies the orthogonal complements of the source subspaces. Transporting the primed blocks to the original spaces gives
\[
\begin{aligned}
    \widetilde A&:=V_K^*A'V_H
        =\Pi_{\mathfrak K}V_K^*U'V_H\Pi_{\mathfrak H},\\
    \widetilde B&:=V_K^*B'V_H
        =\Pi_{\mathfrak K}V_K^*U'V_H(I-\Pi_{\mathfrak H}),
\end{aligned}
\]
understood as operators $\mathfrak H\to\mathfrak K$ and $\mathfrak H^\perp\to\mathfrak K$, respectively. Since
\[
    V_K^*U'V_H-U
    =(V_K^*-I)U'V_H+(U'-U)V_H+U(V_H-I)
\]
is compact, so are $\widetilde A-A$ and $\widetilde B-B$. Compact-perturbation stability of the Fredholm index and unitary equivalence on the source and target spaces give the claimed Fredholmness and index equalities, separately for each Fredholm block.

The compact-perturbation argument also applies when the source and target spaces vary, provided the blocks are compared through unitary identifications. If $A':\mathfrak H'\to\mathfrak K'$ and $V_H:\mathfrak H\to\mathfrak H'$, $V_K:\mathfrak K\to\mathfrak K'$ are unitary with $V_K^*A'V_H-A$ compact, then $\ind A'=\ind A$, and the eigenstate bound follows as in (a). The same holds for $B$. 

Applying (a) to the primed encoding proves the eigenstate lower bounds for every phase sequence and iteration length. At the same phase sequence these eigenvalues agree with the original ones, since the four $\omega$'s depend only on the phases.

To compare the essential spectra, use
\[
    e^{i\alpha(2P-I)}=e^{-i\alpha}I+2i\sin\alpha\,P.
\]
At fixed phases, the differences of the corresponding source and target rotations are therefore compact. Since $U'-U$ and $(U')^*-U^*$ are compact, expanding the differences of the finite products gives $W'_k-W_k$ compact for every $k$, and then ${W'}^{(n)}-W^{(n)}$ compact. Invariance of Fredholmness under compact perturbations, applied to ${W'}^{(n)}-z$ and $W^{(n)}-z$, yields equality of their essential spectra.

\emph{Continuous deformations.} Let the encoding data depend norm-continuously on $t\in[0,1]$, and denote the corresponding blocks by $A_t,B_t$. For each $t_0$, the source and target projections are within distance $1$ of those at $t_0$ when $t$ is sufficiently close to $t_0$. The preceding formula for $V$ then supplies norm-continuous unitaries $V_H(t),V_K(t)$, equal to $I$ at $t_0$, identifying the subspaces at $t_0$ with those at $t$. Consequently,
\[
    V_K(t)^*A_tV_H(t),\qquad V_K(t)^*B_tV_H(t)
\]
are norm-continuous families between fixed source and target spaces. The corresponding transported block remains Fredholm by assumption and unitary equivalence. Local constancy of the Fredholm index therefore makes the corresponding index locally constant in $t$; connectedness of $[0,1]$ makes it constant throughout the deformation. Part (a) gives the claimed eigenstate lower bounds and, at fixed phases, the same eigenvalues.

Finally, the local identifications tend to $I$ as the projections approach their original values, so the transported blocks converge in norm to $A$ and $B$ as the encoding data converge. Openness of the Fredholm set and local constancy of the index on fixed spaces prove stability under sufficiently small norm perturbations of all three encoding data. 
\end{document}